\documentclass[11pt]{article}
\usepackage[utf8]{inputenc}
\usepackage{amsmath,amsfonts,amssymb}
\usepackage{amsthm}
\usepackage{latexsym}
\usepackage{graphicx}
\usepackage{xcolor}
\usepackage{dirtytalk}
\usepackage{mathtools}
\usepackage[margin=1in]{geometry}
\usepackage{thm-restate}
\usepackage{enumitem}
\usepackage[linesnumbered,boxed,ruled,vlined]{algorithm2e}
\usepackage[colorlinks=true, allcolors=blue]{hyperref}
\usepackage[capitalise]{cleveref}
\hypersetup{
    citecolor={violet}
}

\usepackage{mleftright}  %
\let\left\mleft
\let\right\mright

\usepackage{afterpage}

\newtheorem{theorem}{Theorem}[section]
\newtheorem{corollary}[theorem]{Corollary}
\newtheorem{lemma}[theorem]{Lemma}

\newtheorem{claim}[theorem]{Claim}
\theoremstyle{definition}
\newtheorem{definition}[theorem]{Definition}
\newtheorem{remark}[theorem]{Remark}

\newcommand{\eps}{\varepsilon}

\crefname{algocf}{Algorithm}{Algorithms}
\Crefname{algocf}{Algorithm}{Algorithms}
\crefname{claim}{Claim}{Claims}
\Crefname{claim}{Claim}{Claims}

\usepackage{float}

\newfloat{Distribution}{htbp}{loa}
\crefname{Distribution}{Distribution}{Distributions}
\Crefname{Distribution}{Distribution}{Distributions}
\newfloat{Protocol}{htbp}{loa}
\crefname{Protocol}{Protocol}{Protocols}
\Crefname{Protocol}{Protocol}{Protocols}

\SetKwProg{Function}{Function}{:}{}

\SetKwProg{CodeBlock}{}{}{}
\SetKwProg{Repeat}{Repeat}{:}{}

\DeclarePairedDelimiter{\bk}{(}{)}
\DeclarePairedDelimiter{\Bk}{[}{]}
\DeclarePairedDelimiter{\BK}{\{}{\}}

\DeclarePairedDelimiter{\abs}{\lvert}{\rvert}

\DeclarePairedDelimiterX\mysetbase[2]{\lbrace}{\rbrace}{#1\,\delimsize\vert\,#2}
\NewDocumentCommand{\myset}{sO{}m m}{%
  \IfBooleanTF{#1}%
    {\mysetbase*{#3}{#4}}%
    {\mysetbase[#2]{#3}{#4}}%
}

\DeclareMathOperator*{\E}{\mathbb{E}}

\let\Pr\PrAux
\DeclareMathOperator{\poly}{poly}

\renewcommand{\d}{\mathrm{d}}

\newcommand{\defeq}{\coloneqq}

\renewcommand{\l}{\ell}
\renewcommand{\emptyset}{\varnothing}
\renewcommand{\epsilon}{\eps}

\newcommand{\defn}[1]{\emph{\boldmath\textbf{#1}}}

\newcommand{\numberthis}{\addtocounter{equation}{1}\tag{\theequation}}
\newcommand{\smallsub}{\scriptscriptstyle}

\usepackage{regexpatch}
\makeatletter
\xpatchcmd\thmt@restatable{%
\csname #2\@xa\endcsname\ifx\@nx#1\@nx\else[{#1}]\fi
}{%
\ifthmt@thisistheone
\csname #2\@xa\endcsname\ifx\@nx#1\@nx\else[{#1}]\fi
\else
\csname #2\@xa\endcsname[{Restated}]
\fi}{}{}
\makeatother

\usepackage{placeins}

\newcommand{\hBucket}{h_{\textup{buk}}}
\newcommand{\U}{\mathcal{U}}

\newif\ifdraft
\drafttrue

\newif\ifnotanon
\notanonfalse

\title{Tight Bounds and Phase Transitions for Incremental and Dynamic Retrieval}

\author{
William Kuszmaul%
\thanks{Partially supported by a Harvard Rabin Postdoctoral Fellowship and by a Harvard FODSI fellowship under NSF grant DMS-2023528. \texttt{kuszmaul@cmu.edu}.}\\
CMU
\and
Aaron Putterman%
\thanks{Supported in part by the Simons Investigator Award of Madhu Sudan and Salil Vadhan and NSF Award CCF 2152413. \texttt{aputterman@g.harvard.edu}.}\\
Harvard University
\and
Tingqiang Xu%
\thanks{\texttt{xtq23@mails.tsinghua.edu.cn}.}\\
Tsinghua University
\and
Hangrui Zhou%
\thanks{\texttt{zhouhr23@mails.tsinghua.edu.cn}.}\\
Tsinghua University
\and
Renfei Zhou%
\thanks{\texttt{renfeiz@andrew.cmu.edu}.}\\
CMU
}

\date{}

\begin{document}

\maketitle

\begin{abstract} 
    Retrieval data structures are data structures that answer key-value queries without paying the space overhead of explicitly storing keys. The problem can be formulated in four settings (static, value-dynamic, incremental, or dynamic), each of which offers different levels of dynamism to the user. In this paper, we establish optimal bounds for the final two settings (incremental and dynamic) in the case of a polynomial universe. Our results complete a line of work that has spanned more than two decades, and also come with a surprise: the incremental setting, which has long been viewed as essentially equivalent to the dynamic one, actually has a phase transition, in which, as the value size $v$ approaches $\log n$, the optimal space redundancy actually begins to \emph{shrink}, going from roughly $n \log \log n$ (which has long been thought to be optimal) all the way down to $\Theta(n)$ (which is the optimal bound even for the seemingly much-easier value-dynamic setting). 
\end{abstract}

\pagenumbering{arabic}
\thispagestyle{empty}

\section{Introduction}

\SetKwFunction{Insert}{Insert}
\SetKwFunction{Update}{Update}
\SetKwFunction{Query}{Query}
\SetKwFunction{Delete}{Delete}

Given a set of $n$ keys $K \subseteq \mathcal{U}$, and given a $v$-bit value $f(k)$ to be associated with each key $k\in K$, a \emph{retrieval data structure} is a data structure that can answer queries of the form
$$ \texttt{Query}(k) = \begin{cases} f(k) & \text{ if }k \in K \\ \text{anything} & \text{ otherwise.}\end{cases}$$ 
In general, the universe $\mathcal{U}$ of possible keys and the size $v$ of values could be arbitrary -- however, the most common case, and the case we will focus on in this paper, is the setting where $|\mathcal{U}| = n^{1 + \Theta(1)}$ and where $v \le \log n$.\footnote{In fact, the restriction that $v \le \log n$ will not be necessary for our results, but it will simplify the discussion of past upper bounds for the dynamic version of the problem.}

Intuitively, retrieval data structures operate in a manner akin to a promise problem, in the sense that for $k \in K$, $\Query{k}$ must be $f(k)$, but for $k \not\in K$, $\Query{k}$ can be anything. As we shall discuss, this allows for the data structure to use far fewer bits than would be required to store the key value-pairs $\{(k, f(k)) \mid k \in K\}$ explicitly.

In addition to queries, there are three additional operations that a retrieval data structure can support in order to offer varying levels of dynamism: 
\begin{enumerate}
    \item \Update{$k_i, f_i$}: Given some $k_i \in S$, updates the value stored for $k_i$ to be $f_i$. 
    \item \Insert{$k_i, f_i$}: Given $k_i \notin S$, adds $k_i$ to $S$ and stores value $f_i$ for $k_i$. Insertions are only allowed if $|S| < n$. 
    \item \Delete{$k_i, f_i$}: Given $k_i \in S$, removes $k_i$ from $S$. 
\end{enumerate}

A retrieval data structure is said to be \emph{static} if it supports only queries; to be \emph{value-dynamic} if it supports only queries and updates; to be \emph{incremental} if it supports queries, updates, and insertions; and to be \emph{dynamic} if it supports all of queries, updates, insertions, and deletions.

Regardless of the setting, any solution must use at least $nv$ bits. If the solution uses space $nv + R$, for some $R$, then it is said to have \emph{redundancy} $R$. The main task in the study of retrieval data structures is to determine the optimal redundancy $R$ for each of the four settings \cite{DP08, Por09, DW19, KW24, CKRT04, MPP05, DHPP06}. 

In the static and value-dynamic settings, one can use minimal perfect hashing \cite{Meh84, HT01} to achieve a redundancy of $R = \Theta(n)$. However, in the static setting, it turns out that one can do even better: a sequence of works by Dietzfelbinger and Pagh \cite{DP08}, Porat \cite{Por09}, and Dietzfelbinger and Walzer \cite{DW19} show that the optimal redundancy is of the form $o(n)$, and that this can even be achieved with $O(1)$ query time. Perhaps surprisingly, the value-dynamic setting actually requires more space than the static one. This was shown in a recent result by Kuszmaul and Walzer \cite{KW24} who established a tight bound of $R = \Theta(n)$ for any value-dynamic solution. 

The incremental and dynamic settings have also been studied \cite{CKRT04, MPP05, DHPP06}. The state-of-the-art upper bound, by Demaine, Meyer auf der Heide, Pagh, and P\v{a}tra\c{s}cu (Theorem 3 of \cite{DHPP06} with $t = n/v$), is a dynamic constant-time solution that uses space 
$$nv + \Theta(n \log \log n)$$
bits, making for a redundancy of $R = \Theta(n \log \log n)$. The best known lower bound, which is due to Mortensen, Pagh, and P\v{a}tra\c{s}cu \cite{MPP05}, holds for \emph{both} the incremental and dynamic cases, and establishes that any solution must use space at least 
$$\Omega(nv + n \log \log n)$$ 
bits. 

Combined, the upper and lower bounds for the incremental and dynamic cases would seem to suggest that the optimal redundancy should be $R = \Theta(n\log \log n)$ bits. However, the gap between the bounds leaves open an intriguing possibility. It is conceivable that, as $v$ grows, there is an \emph{economy of scale} in which the optimal redundancy $R$ actually \emph{shrinks}. Notice, in particular, that, in the common case where $v = \Omega( \log \log n)$, the known lower bound does not say anything nontrivial about $R$. 

In this paper, we show that such an economy of scale \emph{actually does occur} in the incremental case. We also develop new lower bounds that allow for us to obtain tight asymptotic bounds on $R$ in both the incremental and dynamic cases.

\paragraph{Our results.}
We begin by considering the incremental setting. Our first result is a new upper bound, showing that as $v$ approaches $\log n$ (and also for $v > \log n$), it becomes possible to achieve redundancy significantly smaller than $n \log \log n$. 

\newcommand{\tmpfootnote}{\footnote{Strictly speaking, the space bound of our algorithm is $nv + O(n) + O\bk*{\max\BK*{n \log \bk*{\log(n)/v}, \, 0}}$ bits, as $n \log \bk*{\log(n)/v}$ is negative when $v > \log(n)$. For conciseness, we omit $\max\{\,\cdot\,, 0\}$ for this term in the rest of the paper.}
}
\begin{restatable}{theorem}{introSpaceEfficient}
\label{thm:introSpaceEfficient}
\label{thm:infoTheoryUpperBound}
For universe size $|\mathcal{U}| = \poly(n)$, there is an incremental retrieval data structure that supports all operations using an expected size of $S \leq nv + O(n) + O \left ( n \log \left ( \frac{\log(n)}{v} \right ) \right )$ bits.%
\tmpfootnote
\end{restatable}
\renewcommand{\tmpfootnote}{}

Our initial construction for Theorem \ref{thm:introSpaceEfficient} is existential, and, \emph{a priori}, would seem quite difficult to implement time efficiently. Nonetheless, with a number of additional data-structural ideas, we show that it is also possible to obtain a time-efficient version of the same theorem, with $O(1)$-time queries and updates, and $O(1)$ amortized-expected time insertions.

To internalize Theorem \ref{thm:introSpaceEfficient}, it is helpful to consider several special cases. When $v  = \Omega( \log(n))$, Theorem \ref{thm:introSpaceEfficient} gives the same result as in the value-dynamic case -- the optimal redundancy becomes $\Theta(n)$ bits. For smaller values of $v$, we get a more subtle tradeoff. For example, when $v = \log(n) / \log\log(n)$, we get a redundancy of $R = O(n\log\log\log(n))$. 

Curiously, the theorem suggests a strong phase transition as one transitions from $v \ll \log n$ to $v = \Omega(\log(n))$. For instance, when $v = \log^{0.99}(n)$, the redundancy that we get is the same as in past results, namely, $\Theta(n \log \log n)$.  Our second result is a lower bound showing that, perhaps surprisingly, the tradeoff curve in Theorem \ref{thm:infoTheoryUpperBound} is optimal. 

\begin{theorem}
    Any incremental retrieval data structure with universe size $|\mathcal{U}| \geq n^3$ requires an expected size of $S \geq nv +\Omega(n) + \Omega \left ( n \log \left ( \frac{\log(n)}{v} \right ) \right )$ bits.
    \label{thm:introlowera}
\end{theorem}

Thus, there is necessarily a tight phase transition as the value size goes from slightly below $\log(n)$ bits to slightly above. 

Our final result is a separation between the incremental and dynamic cases. Whereas the incremental case exhibits an economy of scale, we prove that the dynamic case does not. 

\begin{theorem}
    For the dynamic retrieval data structure problem with $|\mathcal{U}| \geq n^3$, any valid data structure requires expected size $S \geq nv +\Omega(n \log\log(n))$ bits.
    \label{thm:introlower}
\end{theorem}

Note that one cannot hope to prove Theorem \ref{thm:introlower} by simply extending the previous lower bound \cite{MPP05}. This is because past lower bounds apply to \emph{both} the incremental and dynamic settings, and Theorem \ref{thm:introlower} holds for \emph{only} the dynamic setting. Instead, Theorem \ref{thm:introlower} is proven by combining the ideas from \cite{MPP05} with an intricate sequence of communication-protocol arguments. What we ultimately show is that, if Theorem \ref{thm:introlower} fails, then it must be possible to construct an impossibly efficient one-way protocol (think of this as an impossibly good compression scheme). This style of proof is similar to the style recently used by Kuszmaul and Walzer \cite{KW24}, although the two settings require very different ideas. 

Combined, our results give a significantly more complete picture for the landscape of retrieval data-structure design. The incremental and dynamic cases, which were previously thought to be essentially equivalent, turn out to behave very differently as $v$ grows. And the incremental case, in particular, exhibits an unexpected economy of scale that allows for a redundancy as small as $\Theta(n)$.

\section{Upper Bounds for Incremental Retrieval}
\label{sec:upper_bound}
\paragraph{Technical overview.}

Our upper bound results rely on the following observation regarding incremental retrieval:
Suppose our target is a data structure using $S$ bits of space (recall $S \geq n \cdot v$). 
If some number $\ell \leq n$ of key-value pairs have been inserted into the data structure so far, then we are free to store an extra $\leq \frac{S - \ell \cdot v}{\ell}$ bits of information for each key that has been inserted. 
We will show that, by using these extra bits of information, we can \emph{decrease} the collision rate for the items that are inserted earlier.
Then, as more and more key-value pairs are inserted, we show how to implement a type of \emph{universe reduction} that decreases the amount of extra information that we store for each key, in order to avoid violating the maximum space of $S$ bits. Importantly, because this universe reduction happens \emph{after} the items have been inserted, we can ensure that this step does not introduce any new collisions. Finally, by performing this universe-reduction step $O(\log^*(n))$ times, we can gracefully adjust the collision rate over time to achieve our upper bound of space usage.

To get our efficient update time, we invoke a series of data-structural techniques. The first barrier is in performing the aforementioned \say{universe-reduction} step. Over large universes, efficiently finding a valid universe-reducing hash function is intractable. By bucketing the keys, and representing the universe-reducing hash function properly, we can employ lookup table techniques to bypass this efficiency barrier. In essence, we can reduce the effective size of the universes we are dealing with to be sub-logarithmic, whereby we can employ explicit lookup tables for finding valid universe reductions. This step of the argument also introduces another problem, which is that the buckets we use are quite small -- too small for Chernoff bounds to be effective -- for which we will apply additional techniques to handle certain bad events.

However, this step only gets us to an amortized, expected update time of $O(\log^{*}(n))$, as there can be as many as $\log^{*}(n)$ universe-reduction steps, each of which requires rebuilding the entire underlying data structure. Here, we break the data structure into two parts: a part in which the universe-reduction hash function is stored (this part uses $o(n)$ machine words that can be efficiently rebuilt using lookup-table techniques), and another component storing the actual values for each key. We connect the two components together using tiny-pointer-like techniques \cite{BCFKT23} (along with other techniques from \cite{BCFKT23}). Critically, these techniques allow us to point each key at its value in a way that \emph{does not have to be rebuilt} over time, so that the total cost of performing rebuilds across the $\log^{*}(n)$ universe-reduction steps is $O(n)$ work.

\subsection{Time-Inefficient Construction}
\label{sec:inefficient}
We start by presenting a time-inefficient data structure for incremental retrieval, but yields the optimal space consumption.

\introSpaceEfficient*

\begin{remark}
We can assume $v \ge \log^{0.99}(n)$ in \cref{thm:infoTheoryUpperBound}, because otherwise, the target space usage becomes $nv + O(n \log \log(n))$, which can be achieved by the best-known dynamic retrieval data structure in \cite{DHPP06}. Further, throughout this section, we will also focus on the most interesting case of $v \le 2^{-10} \log(n) = \Theta(\log(n))$ for the convenience of presentation; the case of larger $v$ is strictly easier and can be solved using our data structure with slight, standard modifications.
\end{remark}

The proof of \cref{thm:infoTheoryUpperBound} will proceed in the following manner:
\begin{enumerate}
    \item Initially, we randomly hash the keys to fingerprints in a universe $[n \log^2(n)]$ (i.e., each fingerprint consists of $\log(n) + 2 \log\log(n)$ bits). For most keys, we will store their fingerprints instead of themselves to save space.
    We insert the fingerprints (and the associated values) of the first $n - \frac{20n \log\log(n)}{v}$ inserted keys into a hash table, with the caveat that if any two keys hash to the same fingerprint, we simply store the key inserted later explicitly in a collision set, costing $O(\log n)$ bits. Because we are inserting only $n - \frac{20n \log\log(n)}{v}$ fingerprint-value pairs, we can afford to ``waste'' at most $20 \log \log(n)$ bits for each key in a sense ``for free'', because we can use the memory saved for the values of the remaining $\frac{20 n \log\log(n)}{v}$ keys as storage. The space cost of storing fingerprints for keys is within this allowed waste of space.
    \item However, after inserting these elements, we are stuck, as the fingerprint sizes are too large and will have filled up the memory used for the remaining keys. So, we employ a technique called \emph{universe reduction}, which decreases the universe size for fingerprints (and hence the fingerprint size) while retaining the fact that there are no collisions between keys. This reduces our space consumption sufficiently such that we may continue inserting another round of keys.
    \item We continue inserting fingerprints of keys until we use roughly $nv$ bits of space, then re-sizing the fingerprint universe until we have inserted all $n$ key-value pairs.
    \item The benefit from this iterative procedure is that many keys are inserted with a \emph{larger fingerprint size} and therefore their probability of collision is smaller, so our table uses less space. After resizing, we then simultaneously get the benefit of a low collision rate \emph{and} a small fingerprint size.
\end{enumerate}

As discussed in the above outline, we make use of the following structure:

\begin{definition}
    For a hash table $H$ on a universe of size $U_1$, we say that a hash function $h: [U_1] \rightarrow [U_2]$ is a \emph{perfect universe reducing hash function} if for every pair of keys $k_1, k_2$ in $H$, $h(k_1) \neq h(k_2)$.
\end{definition}

This condition is equivalent to saying that the hash function $h$ reduces the universe size \emph{without} introducing any collisions. The explicit procedure of our incremental retrieval data structure is shown in \cref{alg:incrementalRetrieval}.

\begin{algorithm}[t]
\caption{IncrementalRetrieval$(u, v, (k_i, v_i)_{i = 1}^n)$}\label{alg:incrementalRetrieval}
\DontPrintSemicolon
    Let $\ell$ be the largest integer such that $\log^{(\ell)}(n) \geq \frac{\log(n)}{v}$. \;
    Initialize $C = \emptyset$ (for collisions). \;
    \For{$j \in [\ell]$} {
        Let $n_{j} = n - \sum_{p = 1}^{j-1} n_p - 20n \frac{\log^{(j+1)}(n)}{v}$. \;
        Let $t_{j} = 2\log^{(j+1)}(n)$.
    }
    Let $n_{\ell+1} = n - \sum_{p = 1}^{\ell} n_p$, and $t_{\l + 1} = 2 \log(\log(n)/v)$. \;
    Initialize $H_0$ to be an empty hash table, and let $t_0 \defeq \log (U/n) \geq 2\log\log(n)$. \;
    \For{$j \in [\ell]$} {
        \label{line:loop}
        Let $h_{j}$ be a perfect universe-reducing hash function of $H_{j-1}$ from $[n \cdot 2^{t_{j-1}}]$ to $[n \cdot 2^{t_{j}}]$. \;
        Apply $h_{j}$ to all fingerprints in $H_{j-1}$, creating a new hash table $H_{j}$ with key universe $[n \cdot 2^{t_{j}}]$. Then delete $H_{j-1}$ from the memory.\footnotemark \label{line:first_reduction}\;
        Insert the fingerprints of the next $n_{j}$ keys into $H_{j}$ with their associated values, where the fingerprints are given by $h_{j}(h_{j-1}(\cdots(h_1(k_i))))$. \;
        If the fingerprint of any inserted key is already present in the hash table, store the key-value pair explicitly in $C$.
    }
    Let $h_{\ell+1}$ be a perfect universe-reducing hash function of $H_\l$ from $[n \cdot 2^{t_{\ell}}]$ to $[n \cdot 2^{2 \log (\log(n)/v)}]$, and let $H_{\ell+1}$ be the result of performing the universe reduction on $H_{\ell}$. \label{line:second_reduction}\;
    Insert the final $n_{\ell+1}$ key-value pairs into $H_{\ell+1}$, storing any collisions explicitly. \label{line:last_round}\;
    \Return{$H_{\ell+1}, h_1, \dots, h_{\ell+1}$, $C$.}
\end{algorithm}

\afterpage{
  \footnotetext{A potential challenge is that, when we reconstruct the hash table in a straightforward way, we need a large piece of temporary memory, unless we apply additional techniques to rebuild the hash table in-place. This issue will be resolved in later subsections.}
}

The process of inserting $n$ key-value pairs in \cref{alg:incrementalRetrieval} is partitioned into $\l + 1$ \emph{rounds}, where the first $\l$ rounds correspond to the $\l$ iterations of \cref{line:loop}, and the last round is shown in \cref{line:last_round}. In round $j \in [\l]$, we insert $n_j$ new key-value pairs to the hash table, followed by a universe reduction on \cref{line:first_reduction} or \cref{line:second_reduction}. Note that \cref{line:first_reduction} has no effect on the first round ($j = 1$) as the input table $H_0$ is empty; one may assume the initial hash function $h_1$ maps the keys to initial fingerprints of length $\log n + 2 \log \log n$. When we analyze universe reductions, we will always assume $j \ge 2$. Moreover, we assume all $t_j$ are integers for simplicity; we also have $t_j \ge 20$ since $t_j \ge 2 \log (\log(n) / v)$ where $\log(n) / v \ge 2^{10}$ by our earlier assumption.

The correctness of \cref{alg:incrementalRetrieval} is clear to see. Indeed, for any query operation, we simply check whether the key $k$ is present in the collision set $C$, and if not, we hash the key using $h_1, \dots, h_{\ell+1}$ to get its fingerprint, and return the value associated with the fingerprint in the hash table. Note that it is necessary to store not only the hash table, but also the universe-reducing hash functions. An update operation is much the same, though once we find the fingerprint-value pair, we would simply update the value. Thus, we spend the remainder of this section analyzing the space complexity. The space consumption comes from three sources: the hash table, the universe-reducing hash functions $h_1, \dots, h_{\ell+1}$, and the collision set $C$. We begin by bounding the space used by the hash table.

\begin{claim}\label{clm:intermediateTableSpace}
    For $j \in [\ell]$, at the end of the $j$th round, the space required to store $H_{j}$ is $\leq nv + O(n)$ bits.
\end{claim}

\begin{proof}
    Recall that in the $j$th round, the universe size we use for the fingerprint is $[n \cdot 2^{2\log^{(j+1)}(n)}]$. Likewise, by the end of the $j$th round, the number of elements that have been inserted into the hash table is $\leq n - 20n \frac{\log^{(j+1)}(n)}{v}$. So, the total space required is
    \[
    \bk*{n - 20n \frac{\log^{(j+1)}(n)}{v}} v + \log \binom{n \cdot 2^{2\log^{(j+1)}(n)}}{n - 20n \frac{\log^{(j+1)}(n)}{v}}.
    \]
In particular, the first term is contributing 
\[
nv - 20 n \log^{(j+1)}(n)
\]
bits, so our goal is to show that 
\[
\log \binom{n \cdot 2^{2\log^{(j+1)}(n)}}{n - 20n \frac{\log^{(j+1)}(n)}{v}} \leq 20 n \log^{(j+1)}(n) + O(n),
\]
as this then leads to an overall bound of $nv + O(n)$.

Using Stirling's approximation, we know
\begin{align*}
&\log \binom{n \cdot 2^{2\log^{(j+1)}(n)}}{n - 20 n \frac{\log^{(j+1)}(n)}{v}} \leq \left (n - 20 n \frac{\log^{(j+1)}(n)}{v} \right ) \cdot \log \bk*{\frac{e n \cdot 2^{2\log^{(j+1)}(n)}}{n - 20 n \frac{\log^{(j+1)}(n)}{v}}} \\
&\leq \left (n - 20 n \frac{\log^{(j+1)}(n)}{v} \right ) \cdot \log(2e \cdot 2^{2\log^{(j+1)}(n)}) \leq 2 n \log^{(j+1)}(n) + O(n)
\end{align*}
as we desire, where the second inequality holds because $n - 20n \frac{\log^{(j+1)}(n)}{v} \geq n/2$.

Thus, the total space required to store $H_j$ for $j \in [\ell]$ is at most $nv + O(n)$.
\end{proof}

Now, we likewise bound the final size required for $H_{\ell+1}$, the final hash table that we return.

\begin{claim}\label{clm:finalTableSpace}
    The space required to store $H_{\ell+1}$ is $\leq nv + O(n + n \log (\log(n) / v))$ bits.
\end{claim}

\begin{proof}
    At the end of the insertions, we have $\leq n$ elements in the hash table. The universe size for fingerprints is $n \cdot 2^{2 \log (\log(n) / v)}$. Thus, we can store all the necessary information with 
    \[
    \log\binom{n \cdot 2^{2 \log (\log(n) / v)}}{n} + nv
    \]
    bits.
    Using Stirling's approximation again, we can bound the first term as 
    \[
    \log\binom{n \cdot 2^{2 \log (\log(n) / v)}}{n} \leq n \cdot \log(2^{2 \log (\log(n) / v)}) + O(n) = 2n \log (\log(n) / v) + O(n)
\]
bits.
\end{proof}

Now, we bound the number of bits required to store the perfect universe-reducing hash functions. 

\begin{claim}\label{clm:intermediateHashSpace}
    In the $j$th round, we require $O(\frac{n}{2^{t_j}})$ bits in expectation to store the universe-reducing hash functions.
\end{claim}

\begin{proof}
    For simplicity, we assume for now that the data structure is given access to a read-many tape of random bits. This assumption will be removed in later subsections when we introduce the time-efficient implementation of our data structure.
    The data structure breaks this read-many tape into groups of random bits, where each group encodes a hash function. Now, the goal of the data structure is to find a perfect, universe-reducing hash function. In order to do this, the algorithm tries each group, one at a time, until it finds one which encodes such a function, and then stores the index of this universe-reducing hash function (i.e., of the group). We will bound the expected size of this index.

    Consider a uniformly random hash function $h:[n \cdot 2^{t_{j - 1}}] \rightarrow [n \cdot 2^{t_j}]$. Given that the hash table $H_{j-1}$ has $\leq n$ elements, we want to calculate the probability that $h$ maps each of the $\leq n$ keys (fingerprints) in $H_{j-1}$ to distinct elements in $[n \cdot 2^{t_j}]$.

    In the worst case, for a given key $k \in [n \cdot 2^{t_{j - 1}}]$, we must hash $k$ such that it avoids the hash values of all $n$ other keys. Using that $h$ is uniformly random, the probability of this occurring is $\geq 1 - \frac{1}{2^{t_j}}$. So, the probability that \emph{all} keys hash to distinct values is at least
    \[
    \bk*{1 - \frac{1}{2^{t_j}}}^{n} = \bk*{1 - \frac{1}{2^{t_j}}}^{(2^{t_j}) \cdot \frac{n}{2^{t_j}}} \geq e^{-\Theta(n / 2^{t_j})}.
    \]
    Thus, the expected number of hash functions we must try is $\leq e^{\Theta(n / 2^{t_j})}$, and by Jensen's inequality, the expected size of the index we must use is bounded by $O(n / 2^{t_j})$.
\end{proof}

\begin{remark}\label{clm:finalHashSpace}
    By the same logic as above, the space required to store the final hash function $h_{\l + 1}$ is bounded by $O\bk[\big]{\frac{n}{2^{2\log (\log(n) / v)}}}$.
\end{remark}

Finally, we analyze the space required to store the entire collision set.

\begin{claim}\label{clm:intermediateCollisions}\label{clm:finalCollisions}
    In expectation (over the hash functions), the space required to store collisions in the $j$th round is $O\left(\frac{n\log(n)}{v \log^{(j)}(n)}\right)$ bits for $j \in [\l]$, and is $O(n)$ bits for $j = \l + 1$.
\end{claim}

\begin{proof}
    Recall that in the $j$th round, a key $k$ is mapped to a fingerprint by composing the perfect universe-reducing hash functions $h_1, \dots, h_j$. The final fingerprint is then $h_j(h_{j-1}(\cdots h_1(k)))$.
    Fixing a single key $k$ that we are inserting, we bound the probability that $k$'s fingerprint collides with an already-existing fingerprint as follows.

    Let $X_\eta$ denote the event that key $k$ collides with existing keys on its round-$\eta$ fingerprint $h_\eta(h_{\eta-1}(\cdots(h_1(k))))$, but not the round-$(\eta - 1)$ fingerprint $h_{\eta-1}(\cdots(h_1(k)))$. For $\eta \le \l$, since the round-$\eta$ fingerprints of existing keys occupy at most $1 / 2^{2 \log^{(\eta + 1)}(n)}$ fraction of the fingerprint universe, we know $\Pr[X_\eta] \le 1 / 2^{2 \log^{(\eta + 1)}(n)}$; for $\eta = \l + 1$, this probability becomes $\Pr[X_{\l + 1}] \le 1 / 2^{2 \log (\log(n) / v)}$ accordingly.

    If the inserted key $k$ collides with other keys on the round-$j$ fingerprints (so that we store it in the collision set), exactly one of $X_1, \ldots, X_j$ occurs. By a union bound, for $j \le \l$, this probability is bounded by
    \begin{align*}
        \Pr[X_1 \lor X_2 \lor \cdots \lor X_j] \le \sum_{\eta = 1}^j \frac{1}{2^{2 \log^{(\eta + 1)}(n)}} = O\bk*{\frac{1}{2^{2\log^{(j+1)}(n)}}}.
    \end{align*}
    Multiplied by the number of keys we insert in the $j$th round -- which is less than
    $\frac{20n\log^{(j)}(n)}{v}$ \footnote{This statement also holds for the first round ($j = 1$), where the number of inserted keys is still less than $\frac{20 n \log^{(1)}(n)}{v} \ge n$.}
    -- and the $O(\log n)$-bit overhead we spend per collision, we know the expected space to store collisions in the $j$th round is bounded by
    \[
    O\bk*{\frac{1}{2^{2\log^{(j + 1)}(n)}}} \cdot O\bk[\bigg]{\frac{n \log^{(j)}(n)}{v}} \cdot O(\log(n)) = O\bk*{\frac{n \log(n)}{v \log^{(j)}(n)}}
    \]
    bits, as desired.

    By the same reasoning, when $j = \l + 1$, the probability of a newly inserted key colliding with existing fingerprints is
    \[
    \Pr[X_1 \lor X_2 \lor \cdots \lor X_{\l + 1}] \le \sum_{\eta = 1}^{\l} \frac{1}{2^{2\log^{(\eta + 1)}(n)}} + \frac{1}{2^{2 \log(\log(n)/v)}} = O\bk*{\frac{1}{2^{2 \log(\log(n)/v)}}}.
    \]
    Multiplied with the number of keys inserted in the $(\l + 1)$st round -- at most $\frac{20 n \log^{(\l + 1)}(n)}{v}$ -- and the $O(\log n)$-bit overhead per collision, the expected space to store collisions in the $(\l+1)$st round is at most
    \begin{align*}
    &\phantom{{}={}} O\bk*{\frac{1}{2^{2\log(\log(n)/v)}}} \cdot O\bk[\bigg]{\frac{n \log^{(\l + 1)}(n)}{v}} \cdot O(\log(n)) \\
    &\le O\bk*{\frac{n \log^{(\l + 1)}(n)}{v(\log(n)/v)^2} \cdot \log(n)} \\
    &= O\bk*{\frac{n \log^{(\l + 1)}(n)}{\log(n) / v}} \\
    &= O(n),
    \end{align*}
    where the last step holds because $\log^{(\l + 1)}(n) < \log(n) / v$ (recall that $\l$ is the \emph{largest} integer for which $\log^{(\l)}(n) \ge \log(n) / v$).
\end{proof}

Now, we can synthesize the above claims to provide an overall bound on the size of the data structure we are achieving.

\begin{proof}[Proof of \cref{thm:infoTheoryUpperBound}]
    First, we bound the space required to store all of the hash functions. In expectation, this is at most
    \[
    O\left (\frac{n}{2^{2\log (\log(n) / v)}} \right ) + \sum_{j = 1}^{\ell} O\left (\frac{n}{2^{2\log^{(\ell+1)}(n)}} \right ) = O(n)
    \]
    bits by \cref{clm:intermediateHashSpace} and \cref{clm:finalHashSpace}.
    Likewise, the space required for storing all of the elements with collisions explicitly is bounded by 
    \[
    O(n) + \sum_{j = 1}^{\ell} O\left(\frac{n\log(n)}{v \log^{(j)}(n)}\right) 
    \]
    bits in expectation via \cref{clm:intermediateCollisions,clm:finalCollisions}. We can simplify this by observing that when $j = \ell$, $\frac{\log(n)}{v \log^{(j)}(n)} \leq 1$ by our definition of the value $\ell$, and that the terms in the sum terms are decreasing exponentially as $j$ goes from $\ell$ to $1$. Thus, 
    \[
    O(n) + \sum_{j = 1}^{\ell} O\left(\frac{n\log(n)}{v \log^{(j)}(n)}\right) \leq O(n) + O(n) = O(n).
    \]

    Finally, via \cref{clm:intermediateTableSpace,clm:finalTableSpace}, we know that the hash table itself never requires more than $nv + O(n \log (\log(n) / v)) + O(n)$ bits at any point during its construction. Thus, in total, the space required for our data structure is bounded by 
    \[
    nv + O(n \log (\log(n) / v)) + O(n),
    \]
    as we desire.
\end{proof}

However, note that as stated the data structure is fundamentally inefficient. Accessing the hash table, finding the hash functions, and evaluating the hash functions are not explicit operations. Thus, in the following subsections we show how we can overcome these challenges with modifications to the above data structure.

\subsection{Improved Efficiency via Bucketing}
\label{sec:bucketing}
One of the key bottlenecks in the above procedure is the implementation of the universe-reducing hash functions. To improve their time efficiency, we change the hash function construction using bucketing.

\paragraph{Constructing hash functions.}

Recall that in the previous subsection, we use a hash function $h_1 : \mathcal{U} \to [n \log^2 n]$ to hash the keys to fingerprints of size $t_1 \defeq 2 \log \log n$. The idea of our new construction is to partition the fingerprint space $[n \log^2 n]$ into $n/B$ buckets for $B \defeq (\log \log n)^{10}$ -- we use two hash functions $\hBucket : \U \to [n / B]$ and $h_1 : \U \to [B \log^2 n]$, both $B$-wise independent, to map each key to a pair $(b, f) \in [n / B] \times [B \log^2 n]$, where $b$ indicates which bucket the key belongs to, and $f$ is the fingerprint of length $(\log B + 2 \log \log n)$ bits within that bucket. The two hash functions $\hBucket$ and $h_1$ take $O(|\U|^\eps) \le O(\sqrt n)$ space to store, and take $O(1)$ time to evaluate \cite{thorup2013simple}. Later in the universe reductions, we use a function $h_j$ to reduce the fingerprints \emph{within buckets} from a universe of $[B \cdot 2^{t_{j - 1}}]$ to $[B \cdot 2^{t_j}]$, without changing the bucket containing each key. The result is that, in round $j$, each bucket is responsible for maintaining a set of fingerprints in $[B \cdot 2^{t_j}]$.

Moreover, among all $n$ keys to be inserted, the expected number of keys residing in any specific bucket equals $B$. By concentration, we show that any bucket is unlikely to contain more than $4B$ keys:

\begin{claim}
\label{clm:lowFullProbability}
For each inserted key $k$, the probability that the bucket containing $k$ has at least $4B$ other keys is $\log^{-\omega(1)}(n)$.
\end{claim}

\begin{proof}
    WLOG, assume $k$ is mapped to the first bucket. Let $X_i$ be the indicator variable of the event that the $i$th inserted key resides in the first bucket. We know $\E[X_i] = B/n$ and $\E\Bk*{\sum_{i = 1}^n X_i} = B$. Moreover, $X_1, \ldots, X_n$ are $B$-wise independent. If the number of keys other than $k$ in the first bucket is at least $4B$, we must have $X_1 + \cdots + X_n \ge 4B$; by the Chernoff bound with limited independence \cite[Theorem 5]{schmidt1995chernoffhoeffding}, this occurs with probability only $e^{-\Omega(B)} = \log^{-\omega(1)}(n)$.
\end{proof}

At any time, if a newly inserted key $k$ is mapped to a bucket already containing $4B$ keys, we insert the key to the collision set instead. Thus, as each bucket always contains at most $4B$ keys (fingerprints), it takes only
\[
\log \binom{B \cdot 2^{t_j}}{4B} = O(B \cdot t_j) = O(\poly \log \log(n))
\]
bits to describe the set of fingerprints in each bucket, which can be encoded within $O(1)$ words and manipulated by lookup tables efficiently.

\paragraph{Efficient universe reductions.}

Since the fingerprint universe within each bucket is so small, universe reductions (within each bucket) can also be made efficient via lookup tables. Assume for $j \ge 2$, we hope to reduce the sizes for a set of fingerprints $\BK{f_1, f_2, \dots, f_s} \subseteq [B \cdot 2^{t_{j - 1}}]$, where $s \le 4B$ is the number of keys in that bucket. We define a function $g$, such that on input $\{f_1, f_2, \dots, f_k\}$, $g$ returns the index of a hash function $h$ such that $h(f_1), h(f_2), \dots, h(f_k)$ are all distinct and each $h(f_i)$ is in the round-$j$ universe $[B \cdot 2^{t_{j}}]$. The function $g$ is implementable by lookup tables because (1) $g$'s input size is bounded by $O(\poly \log \log(n))$ bits, and (2) $g$'s output -- the index for a desired universe-reducing hash function $h$ -- is not too long, as we will show shortly. The latter point is also important for our algorithm to keep space-efficient.

\begin{claim}
    \label{clm:efficientHashFamily}
    In the $j$th round ($j \in [2, \, \l + 1]$), there exists a set $H^{(j)}$ of hash functions from $[B \cdot 2^{t_{j - 1}}]$ to $[B \cdot 2^{t_j}]$ that can be accessed with indices of size $O\bk[\big]{\frac{B}{2^{t_j}}}$ bits, such that for any set of $4B$ keys (fingerprints) in $[B \cdot 2^{t_{j-1}}]$, there exists a hash function $h^{(j)} \in H^{(j)}$ which maps these keys to elements in $[B \cdot 2^{t_j}]$ without any collisions. Such a $H^{(j)}$ can be constructed in $o(\sqrt n)$ time with high probability.
\end{claim}

\begin{proof}
    First, for a fixed set of $4B$ keys in $[B \cdot 2^{t_{j-1}}]$, a random hash function from $[B \cdot 2^{t_{j-1}}] \rightarrow [B \cdot 2^{t_j}]$ will map the keys to distinct elements with probability at least
    \[
    \left(1-\frac{4B}{B\cdot 2^{t_j}}\right)^{4B}=\left(1-\frac{4}{2^{t_j}}\right)^{\frac{2^{t_j}}{4}\cdot \frac{16B}{2^{t_j}}}\geq 4^{-16B/2^{t_j}},
    \]
    where the last inequality holds as $t_j \ge 20$.

    In particular, among all possible sets of $\leq 4B$ keys, after storing $2^{33B / 2^{t_j}}$ such random functions, in expectation we will have covered an
    \[
    1-\left(1-4^{-16B/2^{t_j}}\right)^{2^{33B/2^{t_j}}}\geq 1-\exp\left(\Omega\bk[\big]{2^{B/2^{t_j}}}\right)
    = 1 - n^{-\omega(1)}
    \]
    fraction of those key sets, where the last equality holds since $t_j \le 2\log \log \log (n)$. Since the number of possible key sets with size $\leq 4B$ is at most
    \[4B{B\cdot 2^{t_{j-1}}\choose 4B}=O\left(4B\cdot (B\cdot 2^{t_{j-1}})^{4B}\right)=n^{o(1)}\] 
    (using the fact that $t_{j-1}\leq O(\log\log n)$), these random hash functions will cover every key set with high probability. Thus, the number of hash functions we need to store is less than $2^{33B / 2^{t_j}}$. Taking the logarithm, it implies that the size of the index we must store for the hash function is bounded by
    \[
    O\left(\frac{B}{2^{t_j}}\right).
    \]
    Further, we have constructed the collection of hash functions in time \[O\left(2^{33B/2^{t_j}}\cdot 4B{B\cdot 2^{t_{j-1}}\choose 4B}\cdot 4B\right)=o\bk[\big]{\sqrt n}.\]
    This concludes the proof.
\end{proof}

\cref{clm:efficientHashFamily} indicates that the hash function for round $j$ requires $B / 2^{t_j}$ bits per bucket to store. Multiplied with the number of buckets $n/B$, the space requirement for round-$j$ universe-reducing hash functions is still $O(n / 2^{t_j})$ in total, which is the same as before (see \cref{clm:intermediateHashSpace}) and does not introduce additional space.

Once we have constructed the hash function families, given the index of $h_j$ ($j \ge 2$) and a fingerprint $f$, we can compute $h_j(f)$ in constant time with the help of a precomputed lookup table. Furthermore, given the encoding of a sequence of hash functions $(h_1, h_2, \ldots, h_j)$ and a key $k$, we can also compute $h_j(h_{j-1}(\cdots(h_1(k))))$ in constant time, which is an operation required for keys inserted in the $j$th round, as stated in the following claim.

\begin{claim}\label{clm:tableForHashfn}
With the help of a lookup table of $n^{o(1)}$ words, given the encoding of a sequence of universe-reducing hash functions $(h_1, h_2, \ldots, h_j)$ and a key $k$, we can compute $h_j(h_{j-1}(\cdots(h_1(k))))$ in constant time.
\end{claim}

\begin{proof}
The encoding length of the sequence of universe-reducing hash functions is at most
\[\sum_{j=2}^{\l + 1} O(B / 2^{t_j}) = O(B) = O(\poly \log \log(n))\]
bits. Although the key $k$ itself is too long to serve as the input of a lookup-table function, we take its round-1 fingerprint $f \defeq h_1(k) \in [B \cdot \log^2(n)]$, which takes only $O(\log \log(n))$ bits to represent, and let the lookup table compute $h_j(h_{j-1}(\cdots(h_2(f))))$. The lookup table consists of at most $2^{\poly \log \log(n)} = n^{o(1)}$ words.
\end{proof}

\paragraph{Collision rates.}
The next step is to bound the number of key-value pairs that we insert to the collision set $C$. They consist of two types: (1) the keys that have fingerprint conflicts with other existing keys, and (2) the keys that are mapped to a bucket already containing $4B$ keys. The latter type is already analyzed in \cref{clm:lowFullProbability}, which indicates that the expected number of keys of the latter type is $n / \log^{\omega(1)} n$. The former type can be bounded via the same argument as \cref{clm:intermediateCollisions,clm:finalCollisions},%
\footnote{Although the first hash function $h_1$ is only $B$-wise independent (not truly random), the calculations in \cref{clm:intermediateCollisions,clm:finalCollisions} still hold because they only require the hash function to be pairwise independent.}
which implies that the total number of collisions throughout all $\l + 1$ rounds is at most $O(n / \log n)$ in expectation. Each key-value pair stored in the collision set costs $O(\log n)$ bits of space, so the expected space consumption of the collision set is $O(n)$ bits, which fits our need.

\paragraph{Putting pieces together.}
So far, we have discussed how to construct hash functions supporting fast evaluation and fast universe reduction. Based on them, we proceed to describe our incremental retrieval data structure with an improved running time of $O(\log^*(n))$ per operation.

\begin{lemma}
For $|\U| = n^{1 + \Theta(1)}$ and $v \in [\log^{0.99}(n), \, 2^{-10} \log(n)]$, there is an incremental data structure which uses $nv + O(n \log(\log(n) / v))$ bits of space and takes $O(\log^*(n))$ time per insertion. Updates and queries take $O(1)$ time.
\end{lemma}

\begin{proof}
The framework of the data structure is still based on \cref{alg:incrementalRetrieval}, despite that we partition keys to buckets and use the hash functions introduced above for universe reductions. In the $j$th round ($j \in [\l + 1]$), we encode the set of fingerprints in each bucket using at most $\log \binom{B \cdot 2^{t_j}}{4B} \le 8 B t_j$ bits, and store these sets as bit strings of a fixed length $8 B t_j$ regardless of the actual number of keys in each bucket, for the convenience of memory management. These sets occupy $\le 8 n t_j$ bits of space in total.

\newcommand{\nlej}[1][j]{n_{\smallsub \le #1}}

In addition, we store the associated values for all keys using a hash table that allows searching for the associated value given the pair $(b, f)$ of a key, where $b$ is the index of the bucket containing the given key and $f$ is the fingerprint of the key in the current round $j$. The \say{key universe} of this hash table is $[n / B] \times [B \cdot 2^{t_j}] \cong [n \cdot 2^{t_j}]$, thus the space usage of the hash table is
\[
\log \binom{n \cdot 2^{t_j}}{\nlej} + \nlej v + O(n)
\numberthis \label{eq:value-table-space}
\]
bits, where $\nlej \defeq n_1 + \cdots + n_j$ is the maximum number of inserted keys in the $j$th round, assuming we use the optimal succinct hash table from \cite{BFKKL22} with $t = \log^*(n)$:

\begin{theorem}{\normalfont\cite{BFKKL22}}\label{thm:tableSpace}\label{thm:hash-table}
Assume $|\U| = \poly(n)$. For any parameter $t \in [\log^*(n)]$, there is a hash table that stores $n$ keys from the universe $\U$ with $v$-bit associated values, uses
\[
\log \binom{|\U|}{n} + nv + O(n \log^{(t)}(n))
\]
bits of space, and supports insertions, deletions, updates, and queries in $O(t)$ expected time. Further, the updates and queries take $O(1)$ time in the worst case.
\end{theorem}

The space consumption of the hash table \eqref{eq:value-table-space}, together with the $8 n t_j$ bits used by the fingerprint sets, can be bounded similarly to \cref{clm:intermediateTableSpace,clm:finalTableSpace}: When $j \le \l$, the space usage is at most
\begin{align*}
& \phantom{{}={}} \log \binom{n \cdot 2^{t_j}}{n} + \nlej v + O(n) + 8 n t_j \\
& \le 2 n t_j + \bk*{n - 20n \frac{\log^{(j + 1)}(n)}{v}} v + 8 n t_j + O(n) \\
& = 2 n t_j + (nv - 10 n t_j) + 8 n t_j + O(n) \\
& = nv + O(n)
\end{align*}
bits (recalling $t_j \defeq 2 \log^{(j + 1) n}$ for $j \le \l$); when $j = \l + 1$, the space usage is bounded by
\[
\log \binom{n \cdot 2^{t_{\l + 1}}}{n} + nv + O(n) + 8 n t_{\l + 1} = nv + O(n \log(\log(n)/v))
\]
bits.
Note that although the hash table stores the set of key fingerprints again, this waste of space is acceptable according to the above calculation.

The third component of our data structure is the collision set, where we explicitly store $O(n / \log(n))$ key-value pairs that encountered two types of collisions, paying $O(\log(n))$ bits of space for each of them. This is also implemented using a hash table in \cref{thm:hash-table} with parameter $t = 1$, occupying an expected space of $O(n)$.

When universe reductions happen, we use lookup tables to find the index of a collision-free hash function for each bucket (see \cref{clm:efficientHashFamily}) and store these indices in the memory. For each bucket, we need to store the $\l$ universe-reducing hash functions $h_1, \ldots, h_\l$ in a consecutive, fixed-sized piece of memory, so that we can read the encoding of them in $O(1)$ time and use a lookup table to support constant-time evaluation as we introduced earlier. Updating the fingerprint sets and hash functions in all buckets only takes $O(n / B)$ time during each universe reduction; however, we also need to rebuild the entire hash table for associated values, because the ``current fingerprints'' for all keys have changed. The latter step costs $O(n)$ time per universe reduction, and since we have at most $\log^*(n)$ universe reductions during all $n$ insertions, this sums up to an expected amortized time of $O(\log^*(n))$ per insertion, which is a bottleneck of our data structure.

We conclude the proof by analyzing the time and space costs of other parts: When we insert a new key-value pair without triggering the universe reduction, we spend $O(1)$ time to update the fingerprint set in a bucket, and spend $O(\log^*(n))$ time to insert an item to either the hash table for associated values or the collision set. Updates and queries are faster: we only need to query or update the hash tables, which take $O(1)$ worst-case time.
For the space cost, in addition to the fingerprint sets, hash tables, and hash functions discussed above, we also need $O(\sqrt n)$ space to store the initial hash functions $\hBucket$ and $h_1$, and $n^{o(1)}$ space to store lookup tables. Both terms are negligible compared to the target redundancy.
\end{proof}

The time bottlenecks of our algorithm are (1) to rebuild the hash table between rounds, and (2) the $O(\log^*(n))$ insertion time for the succinct hash table. The next subsection will show how to change the hash table design and efficiently reconstruct the hash table to achieve the optimal $O(1)$ amortized running time.

\subsection{Improved Efficiency via Linear Probing and Tiny Pointers}
\label{sec:linearprobing}
\label{sec:hashtable}

\newcommand{\hVal}[1][j]{h^{(#1)}_{\textup{val}}}

The key idea to avoid rebuilding the entire hash table between phases is to use the \emph{original keys} instead of the fingerprints as the keys to the hash table for values while avoiding storing the information of the original keys.

Specifically, for the $n_j$ key-value pairs that we will insert during the $j$th round ($j \in [\l + 1]$), we allocate an array of $n_j$ slots, each slot capable of storing a $v$-bit value. When a new key-value pair $(k, y)$ is inserted, we use the \emph{linear probing} algorithm to store $y$ in the first empty slot after $\hVal(k)$, where $\hVal : \U \to [n_j]$ is the hash function used for the linear probing. If the value is stored at the slot $\hVal(k) + p_k$, we say the integer $p_k \ge 0$ is the \defn{offset} of the key $k$. Since the array does not contain any information about the keys, when we retrieve the value associated with a given key $k$, it is necessary to also provide the offset $p_k$ to the data structure. This idea combined with various data-structural techniques leads to the following theorem.

\begin{restatable}{theorem}{LinearProbing}
\label{thm:linearProbing}
Let $\log(n) \le m \le n$ and $v \le \log(n)$.
There is an insertion-only data structure that stores up to $m$ key-value pairs, takes at most $mv + O(m)$ bits of space, and supports the following operations:
\begin{enumerate}
    \item \Insert{$k, y$}: Given a key-value pair $(k, y)$, the hash table stores the value $y$ for $k$, and returns an \emph{offset} $p_k$.
    \item \Update{$k, p_k, y$}: Given an existing key $k$ and its offset $p_k$, the data structure updates $k$'s associated value to $y$.
    \item \Query{$k, p_k$}: Given an existing key $k$ and its offset $p_k$, the data structure returns the value stored for $k$.
\end{enumerate}
All three types of operation take $O(1)$ expected time. Moreover, $\mathbb{E}[\sum \log p_k] \leq O(m)$; the expected number of offsets exceeding $\log^{100}(n)$ is bounded by $O(m / \poly \log(n))$.
\end{restatable}

We defer the proof of \cref{thm:linearProbing} to \cref{sec:insertion_only_linear_probing}, and continue to construct time-efficient retrieval data structures for now.

\smallskip

Previously, we stored the set of fingerprints for each bucket by encoding the information into less than a word and using lookup tables to access/update. Now, we do the same for the offsets $p_k$: \cref{thm:linearProbing} shows that most offsets can be encoded using $O(\log \log(n))$ bits, thus all offsets within a bucket can be encoded within $O(B \cdot \log \log(n)) = \poly \log \log(n)$ bits.
In addition to these offsets, we also store a piece of \emph{round information} $i_k$ for each key $k$, representing which round $k$ gets inserted in. We then modify \cref{alg:incrementalRetrieval} to \cref{alg:incrementalRetrievalUpdated}.

\smallskip

\begin{algorithm}[ht]
\caption{IncrementalRetrieval$(u, v, (k_i, v_i)_{i = 1}^n)$}\label{alg:incrementalRetrievalUpdated}
\DontPrintSemicolon
    Let $\ell$ be the largest integer such that $\log^{(\ell)}(n) \geq \frac{\log(n)}{v}$. \;
    Initialize $C = \emptyset$ (for collisions). \;
    \For{$j \in [\ell]$} {
        Let $n_{j} = n - \sum_{p = 1}^{j-1} n_p - 20n \frac{\log^{(j+1)}(n)}{v}$. \;
        Let $t_{j} = 2\log^{(j+1)}(n)$.
    }
    Let $n_{\ell+1} = n - \sum_{p = 1}^{\ell} n_p$, and $t_{\l + 1} = 2 \log(\log(n)/v)$. \;
    Let $t_0 \defeq \log (U/n) \geq 2\log\log(n)$. \;
    Initialize $K_1, K_2,\dots, K_{n/B}$, $P_1, P_2,\dots, P_{n/B}$, and $I_1, I_2,\dots, I_{n/B}$ to be the fingerprints, offsets, and the round information of the keys in each bucket, respectively. \;
    \For{$j \in [\ell]$} {
        Let $h_{j}$ be a perfect universe-reducing hash function of $H_{j-1}$ from $[n \cdot 2^{t_{j-1}}]$ to $[n \cdot 2^{t_{j}}]$. \;
        Apply $h_{j}$ to all fingerprints in $K_1, K_2, \dots, K_{n/B}$ using lookup tables. \;
        Update $P_1, P_2, \dots, P_{n/B}$ and $I_1, I_2, \dots, I_{n/B}$ using lookup tables. \;
        Creating a new linear probing hash table $H_{j}$ with size $n_j$. \;
        \For{key-value pair $(k, y)$ in the next $n_j$ keys} {
            Insert the key-value pair into $H_{j}$, receiving its offset $p_k$. \;
            Find its bucket $b$. \; 
            If the fingerprint of the inserted key is already present in the hash table, or the number of keys in bucket $b$ already achieves $4B$, or $\log p_k\geq 100\log\log (n)$, store the key-value pair explicitly in $C$. \;
            Otherwise, insert it into bucket $b$ and update $K_b, P_b, I_b$.
        }
    }
    Let $h_{\ell+1}$ be a perfect universe-reducing hash function of $H_\l$ from $[n \cdot 2^{t_{\ell}}]$ to $[n \cdot 2^{2 \log (\log(n)/v)}]$, and let $H_{\ell+1}$ be the result of performing the universe reduction on $H_{\ell}$. \;
    Insert the final $n_{\ell+1}$ key-value pairs into $H_{\ell+1}$ similarly to above, storing any collisions/exceptions explicitly. \;
    \Return{$H_{\ell+1}, h_1, \dots, h_{\ell+1}$, $C$.}
\end{algorithm}

In \cref{alg:incrementalRetrievalUpdated}, when storing $P_1, P_2, \dots, P_{n/B}$ and $I_1, I_2, \dots, I_{n/B}$, we concatenate all $p_k$ and $i_k$ in the increasing order of fingerprints, and update them during each insertion.

When we perform updates or queries, we know the original representation of the key temporarily from the input to the operation. We can then use the original key together with its $i_k$ and $p_k$ to update/query the value in the linear-probing hash table.

One remaining challenge is that, the offsets $p_k$ and round information $i_k$ of different keys have different sizes, and so do $P_b$ and $I_b$. As we insert new elements into the data structure, the sizes of $P_b$ and $I_b$ change accordingly, but they never exceed one machine word of $\Theta(\log n)$ bits. In order to allocate memory for $P_b$ and $I_b$ dynamically, we adopt the tool from \cite{BCFKT23}. It shows a hash table maintaining variable-size values within a single word, which supports updates/queries with constant time, and consumes an additional space of $\log^{(r)}(n) + O(\log s)$ bits per $s$-bit value for an arbitrary constant $r > 0$ of our choice. Here we choose $r = 2$, thus wasting $O(\log \log(n))$ bits per bucket.

\paragraph{Memory usage.}

Compared to the previous subsections, the memory usage of \cref{alg:incrementalRetrievalUpdated} has the following parts as the overhead:
\begin{enumerate}
\item The keys with a large offset $\log p_k \ge 100 \log \log(n)$ are stored explicitly in the collision set, costing $O(\log n)$ bits per item. According to \cref{thm:linearProbing}, the expected number of such elements is bounded by $O(n / \poly \log(n))$, so this overhead is $o(n)$ bits.
\item The space to store the offsets $P_b$. \cref{thm:linearProbing} shows that the average offset size is $O(1)$ among all $n$ keys, so this overhead is bounded by $O(n)$ bits.
\item\label{item:i_b} The space to store the round information $I_b$.
\item The space overhead of using \cite{BCFKT23} to dynamically allocate memory for $P_b$ and $I_b$. As discussed, this tool only wastes $O(\log \log(n))$ bits per bucket, which counts for $O((n/B) \log\log(n)) = o(n)$ bits.
\item The space overhead of \cref{thm:linearProbing}, which is only $O(1)$ bits per key and sums up to $O(n)$ bits.
\item Lookup tables. Since all lookup tables we use have input size $O(\poly\log\log(n))$ and output size $O(\log(n))$, it is clear that they occupy $o(n)$ bits of space and $o(n)$ time to precompute.
\end{enumerate}

The space usage of \cref{item:i_b} is bounded by the following claim.

\begin{claim}\label{clm:PandI}
  The encoding length of $I_1, I_2, \ldots, I_{n/B}$ is $O(n)$.
\end{claim}

\begin{proof}
  At least $n-\frac{n\log\log(n)}{v}$ keys are inserted in the very first round, for each of which we only spend $O(1)$ bits to store $i_k$, and the number of rounds is at most $\log^*(n)$, so the sum of all $i_k$ is less than \[n+\frac{n\log\log(n)}{v}(\log^*(n)-1)=O(n),\] completing the proof.
\end{proof}

Based on the above discussion, we are now ready to prove our final result.

\begin{theorem}
    For $|\U| = n^{1 + \Theta(1)}$ and $v \in [\log^{0.99}(n), \, 2^{-10} \log(n)]$, there is an incremental data structure which uses $nv + O(n \log(\log(n) / v))$ bits of space and takes $O(1)$ amortized-expected time per operation.
\end{theorem}

\begin{proof}
    We follow \cref{alg:incrementalRetrievalUpdated}. The above discussion already shows that the space usage of \cref{alg:incrementalRetrievalUpdated} is only $O(n)$ bits more than the previous subsection, thus fitting our desired bound.

    During each insertion, we update the set of fingerprints $K_b$, the offsets $P_b$, and the round information $I_b$ using lookup tables in constant time; we insert the value into one of the linear-probing hash tables, which takes $O(1)$ amortized-expected time by \cref{thm:linearProbing}; in case of collisions/exceptions, we insert the key-value pair into the collision set, which also takes $O(1)$ time by \cref{thm:hash-table}. Updates and queries are easy to perform and take $O(1)$ time in the worst case.

    During a universe reduction, we only need to rewrite the information $K_b, P_b, I_b$ within each bucket, but not rebuild any hash table. Thus, each universe reduction takes $O(n / B)$ time. Multiplied with the number of rounds $O(\log^*(n))$, the amortized time cost of universe reductions is $o(1)$ per operation, which fits in our desired bound.
\end{proof}

\begin{remark}
    Although it is not the main focus of this paper, we can also support larger values $v \ge \Theta(\log(n))$ by adjusting our algorithm: We store each key's offset and round information in the collision set instead of the value, costing $O(\log (n))$ bits of space per key instead of $O(v)$. As we analyzed before, the number of keys getting into the collision set is at most $O(n / \log (n))$, and their space usage does not exceed $O(n)$. We should also adjust the proof of \cref{thm:linearProbing}, making it space-efficient for large $v$'s by eliminating free slots via the standard \say{backyard} technique. Then the total memory usage is still $nv + O(n)$.
\end{remark}

\section{Lower Bound for Incremental Retrieval Data Structures}
\label{sec:lb_incremental}
In this section, we present a lower bound on the size of data structures in the incremental retrieval setting when the value sizes are \emph{small}. In particular, recall that our upper bound allows us to waste only $O(n)$ bits when the value approaches $\log(n)$. Thus naturally, one may wonder if there is a tradeoff when the value is $o(\log(n))$. This is what we establish in this section. 

Let $v$ denote the number of bits used in the values. As in prior works \cite{MPP05}, we reduce from a set distinction problem, though we require very different techniques.
The reduction is based on the following random process:

\newcommand{\RA}{R_A}
\newcommand{\RB}{R_B}
\newcommand{\RAp}{R'_{A'}}
\newcommand{\PB}{P_B}
\newcommand{\Plausible}{(A, \RA)_i^p}

\begin{enumerate}
\item We break the universe $\mathcal{U} = [u]$ into $m = n - n/v$ equal-sized parts, denoted by $\mathcal{U}_1, \dots, \mathcal{U}_m$.
\item Let $A$ be a set formed by sampling one element uniformly at random from each part $\mathcal{U}_i$ ($i \in [m]$).
\item Given the key set $A$, let $\RA$ be random $v$-bit values associated with the keys in $A$, where we require that each value starts with a $0$.
$\RA$ can be encoded using $m$ $v$-bit strings, where the $i$th string is the associated value for the $i$th smallest key in $A$. Hence, we may also refer to the random variable $\RA$ without knowing $A$.
\item We construct a retrieval data structure $F$ by inserting keys in $A$ with values $\RA$.
\item Uniformly sample $n/v$ parts out of the $m$ parts, denoted by $\PB \subseteq [m]$ ($|\PB| = n / v$). Then, we let $B$ be a random set obtained by sampling an element from each of the $n/v$ parts in $\PB$ while avoiding elements in $A$ (i.e., $A \cap B = \emptyset)$.
\item Let $\RB$ be random $v$-bit values associated with the keys in $B$ where each value starts with a 1. We construct the data structure $G$ by inserting $B$ with values $\RB$ into $F$.
\end{enumerate}

Note that the operations we perform on the retrieval data structures are generated according to a fixed distribution independent of the algorithm, so by Yao's minimax principle, we may assume in the proof that the data structure is \emph{deterministic}, denoted by two functions $F = F(A, \RA)$ and $G = G(F, B, \RB)$ that construct the data structures. We will use $H(X)$ to denote the entropy of $X$ over the distribution induced by the random process above. Note that $H(F)$ and $H(G)$ will lower bound the expected space usage of the retrieval data structure.

\begin{remark}
    From now on, we will allow queries to the data structure that return for an element $x$ whether $x \in A$ or $x \in B$ (or otherwise return anything). This is because we can recover the first bit of the value associated with $x$, and use this to discern which set it belongs to.
\end{remark}

\begin{claim}\label{clm:protocolBound}
$H(G) + H(B, \RB \mid F, G) \geq H(B) + nv - O(n)$.
\end{claim}

\begin{proof}
Consider the following communication game. Alice sends $G$, $A$, and $(B, \RB \mid F, G)$ to Bob; Bob then recovers $A$, $\RA$, $B$, and $\RB$ based on the following procedure:
\begin{enumerate}
\item Using $G$ and $A$, Bob recovers $\RA$ by querying the retrieval data structure.
\item Using $A$ and $\RA$, Bob can reconstruct $F$, since $F$ only depends on $A$ and $\RA$.
\item Since now Bob knows both $F$ and $G$, based on $(B, \RB \mid F, G)$, he can recover $B$ and $\RB$.
\end{enumerate}

The entropy of the information Bob recovers is
\begin{align*}
    H(A, B, \RA, \RB) &= H(A) + H(B \mid A) + H(\RA) + H(\RB) \\
    &= H(A) + (H(B) - O(n)) + m(v - 1) + \frac{n}{v}(v - 1) \\
    &= H(A) + H(B) + nv - O(n)
\end{align*}
bits, where the second equality holds because every element in $(B \mid A)$ is uniformly random from $(\frac{u}{m} - 1)$ possibilities, lowering the entropy by at most 1 bit compared to $B$ without conditioning on $A$. This forms a lower bound on the encoding length of Alice's message, i.e.,
\begin{align*}
    H(G) + H(A) + H(B, \RB \mid F, G) &\ge H(A) + H(B) + nv - O(n),
\end{align*}
which implies the claim.
\end{proof}

Rearranging the terms in \cref{clm:protocolBound}, we see $H(G) \ge H(B) - H(B, \RB \mid F, G) + nv - O(n)$, which is a lower bound on the expected size of the data structure $G$. In order to establish the desired lower bound on $H(G)$, it suffices to upper bound $H(B, \RB \mid F, G)$, which is done by creating an efficient encoding scheme.

We start by introducing the notion of plausibility, which captures the set of possible keys in the $i$th part after one sees the data structure $F$.

\begin{definition}
    For a set $A$ (of the prescribed form), a set of the corresponding values $\RA$, and an index $i \in [m]$, we define the plausible set $(A, \RA)_i^p$ as follows: A pair $(a, r_a) \in (A, \RA)_i^p$ if there exists a set $A'$ and associated values $R'_{A'}$ such that $a$ is the $i$th element in $A'$ while its associated value in $\RAp$ is $r_a$, and there is $F(A, \RA) = F(A', \RAp)$.
\end{definition}

\begin{claim}\label{clm:plausibleSize}
    Let $S \ge H(F(A, \RA))$ be the expected size of the data structure $F$, then
    \[
    \E_{i, A, \RA}\Bk[\big]{|(A, \RA)_i^p|} \geq \frac{(u/n)2^v}{2^{S/m + O(1)}}.
    \]
\end{claim}

\begin{proof}
    We prove this via an encoding argument. Suppose Alice selects a random set $A$ and values $\RA$, computes the data structure $F(A, \RA)$, and sends the data structure $F$ to Bob. Once Bob receives the data structure, he can compute the plausible set $(A, \RA)_i^p$. Then, to let Bob recover $A$ and $\RA$, it suffices for Alice to send the index of $(a_i, r_{a_i})$ among all pairs in $(A, \RA)_i^p$, where $a_i$ denotes the $i$th element in $A$ and $r_{a_i}$ denotes its associated value in $\RA$.

    The expected size of the message Alice sends is
    \[
    S + \sum_{i=1}^{m} \E_{A, \RA}\Bk[\big]{\log\abs*{(A, \RA)_i^p}}
    \]
    bits, and the entropy of the information Bob recovers is at least
    \[
    H(A, \RA) = m \log \bk*{\frac{u}{m}} + m(v - 1),
    \]
    so we get the relationship that
    \[
     S + \sum_{i=1}^{m} \E_{A, \RA}\Bk[\big]{\log\abs{(A, \RA)_i^p}} \geq mv + m \log \bk*{\frac{u}{m}} - O(m).
    \]
    This implies
    \[
    \E_{i, A, \RA}\Bk[\big]{\log\abs{(A, \RA)_i^p)}} \geq v + \log\bk*{\frac{u}{m}} - \frac{S}{m} - O(1).
    \]
    By Jensen's inequality,
    \[
    \E_{i, A, \RA}\Bk[\big]{\abs{(A, \RA)_i^p}} \geq \frac{(u/m)2^v}{2^{S/m + O(1)}}.
    \qedhere
    \]
\end{proof}

Next, we extend the definition to \emph{feasibility}, which describes the set of possible elements after one sees both the data structure $F$ and the set $B$.

\begin{definition}
    For $(A, \RA)$ and a pair $(x, r_x) \in (A, \RA)_i^p$, we say that $(x, r_x)$ \emph{remains feasible for $A$ after $B$} if there exists a set $A'$ with values $R'_{A'}$ such that $F(A', R'_{A'}) = F(A, \RA)$, 
    $x$ is the $i$th element in $A'$ with associated value $r_x$ in $\RAp$,
    and $A' \cap B = \emptyset$.
\end{definition}

The only difference between \emph{feasibility} and \emph{plausibility} is that feasibility requires $A' \cap B = \emptyset$.
For a random choice of $B$, we show that in expectation, most elements in $(A, \RA)_i^p$ remain feasible.

\begin{claim}\label{clm:probabilityFeasible}
    Let $A, \RA$ be fixed, and consider an element $(x, r_x) \in (A, \RA)_i^p$. Over a random choice of the set $B$, $(x, r_x)$ remains feasible for $A$ with probability $1 - m^2 / u$.
\end{claim}

\begin{proof}
    Because $(x, r_x) \in (A, \RA)_i^p$, there exists a set $A'$ with corresponding values $\RAp$ such that $F(A, \RA) = F(A', \RAp)$, and $x$ serves as the $i$th element in $A'$ with value $r_x$ in $R'_{A'}$. If we also have $A' \cap B = \emptyset$, then it implies that $(x, r_x)$ remains feasible. Since $B$ randomly samples one element from each of the $n/v$ chosen parts, while in each part the probability of colliding with the element in $A'$ is $1 / (\frac{u}{m} - 1)$, by a union bound, we see that the probability of $A' \cap B \ne \emptyset$ is at most
    \[
    \frac{n}{v} \cdot \frac{1}{\frac{u}{m} - 1}
    \le \frac{m}{2} \cdot \frac{2m}{u}
    = \frac{m^2}{u}.
    \qedhere
    \]
\end{proof}

We remark that the conclusion of \cref{clm:probabilityFeasible} still holds conditioned on $b_i \in B$, where $b_i \in \U_i$ is an arbitrarily given element in the $i$th part of the universe such that $b_i \ne x$ (other elements of $B$ are still chosen randomly). This is by the same proof as above.
When $b_i = x$, any set containing $x$ will intersect with $B$, thus $(x, r_x)$ is impossible to remain feasible.

\begin{claim}\label{clm:expectedFeasible}
    Let $A, \RA$ be fixed. The expected number of pairs $(x, r_x) \in (A,\RA)_i^p$ which do not remain feasible for $A$ after a random choice of $B$ is bounded by $m + 1$. Furthermore, this holds even if $B$ is randomly chosen conditioned on that a fixed element $b_i \in \U_i$ appears in $B$.
\end{claim}

\begin{proof}
First, there cannot be two tuples $(x, r_x)$ and $(x, r_x')$ sharing the same $x$ in $(A, \RA)_i^p$, because querying $x$ on the data structure $F$ will return the unique value $r_x$ associated with $x$. This directly implies that the number of tuples $(x, r_x) \in \Plausible$ is bounded by $u / m$ (the number of elements in each part $\U_i$).

Then, over a random choice of $B$, we know that each such tuple remains feasible for $A$ with probability $1 - m^2 / u$. Thus, the number of tuples which \emph{do not} remain feasible for $A$ is bounded by $\frac{m^2}{u} \cdot \frac{u}{m} = m$ in expectation.

When $b_i \in B$ is given for a fixed element $b_i \in \U_i$, there are two cases for each tuple $(x, r_x) \in \Plausible$: If $x \ne b_i$, then $x$ still has at least $1 - m^2 / u$ probability of remaining feasible after $B$; otherwise, for $x = b_i$, the pair $(x, r_x)$ cannot remain feasible after $B$, but there is only one tuple in the latter case. Taking a summation of the possibilities of not remaining feasible, we know the total number of tuples not remaining feasible is bounded by $\frac{m^2}{u} \cdot \frac{u}{m} + 1 = m + 1$ by expectation.
\end{proof}

Recall that we are allowed to query the data structure $G$ whether a key $x$ comes from the set $A$ or $B$, since this information is encoded using the first bit of the associated values. This type of query is the key to an efficient encoding of $(R, \RB \mid F, G)$. We show the following claims.

\begin{claim}\label{clm:returnA}
    Let $(x, r_x)$ be a tuple that remains feasible for $A$ after $B$. Then, the data structure $G$ must return $A$ when queried with $x$. 
\end{claim}

\begin{proof}
    Because $(x, r_x)$ remains feasible for $A$ after $B$, this means instead of initializing the data structure $F$ with $A, \RA$, we could have equivalently initialized it with $A', R'_{A'}$, where $x \in A'$ has associated value $r_x$ in $\RAp$. Importantly, $A' \cap B = \emptyset$, so this still satisfies our requirement about the disjointness of $A$ and $B$ in the aforementioned random process. Finally, when we then create the data structure $G(F, B, \RB)$, it is impossible to distinguish between the case when $F$ is initialized with $A'$ or $A$, so $G$ must always return that $x$ is in $A$ when queried with $x$.
\end{proof}

\begin{claim}\label{clm:returnB}
    Let $F = F(A, \RA)$ be the intermediate data structure, and $(A, \RA)_i^p$ be the plausible set. Then, over a random choice of $B$, the expected number of tuples $(x, r_x) \in (A, \RA)_i^p$ for which the data structure $G$ can return $B$ when queried with $x$ is at most $m + 1$. This also holds conditioned on that a fixed element $b_i \in \U_i$ appears in $B$.
\end{claim}

\begin{proof}
    We know there are at most $u/m$ possible tuples in $(A, \RA)_i^p$ (one for each choice of $x$). Now, for each such tuple, if $(x, r_x)$ remains feasible for $A$ after $B$, then the data structure must return $A$ when queried with $x$. So, in expectation, the number of $x$'s for which the data structure can return $B$ is bounded by $m + 1$ by \cref{clm:expectedFeasible}.
\end{proof}

Applying the law of total probability on the conclusion of \cref{clm:returnB}, we know that conditioned on that $B$ contains an element in the plausible set $\Plausible$, the expected number of elements in $\Plausible$ that lets $G$ return $B$ is bounded by $m + 1$. This will be used for the following claim.

\begin{claim}\label{clm:encodingSize}
    Let $S \ge \max\BK{H(F), H(G)}$. Given $F$ and $G$, there exists an encoding of $(B, \RB)$ of expected size
    \[
    O(n) + \frac{n}{v} \log \bk*{\frac{u}{m}} - \frac{n}{v} \cdot \frac{2^v}{2^{S/m + O(1)}} \cdot \log\bk*{\frac{u}{m^2}}.
    \]
\end{claim}

\begin{proof}
    Recall that in order to sample $B$, we choose a random subset of $n/v$ parts of the domain, $\mathcal{U}_1, \dots \mathcal{U}_m$, and sample a random element from each such that it does not collide with the set $A$. 
    
    We consider the following encoding argument: 
    \begin{enumerate}
        \item First, we send $m = O(n)$ bits indicating which of the $m$ parts $\mathcal{U}_1, \dots, \mathcal{U}_m$ $B$ samples from.
        \item Then, we send $m = O(n)$ bits indicating whether for $i \in [m]$, there exists $r_{b_i}$ such that $(b_i, r_{b_i}) \in (A, \RA)_i^p$, where $b_i$ is the element in $B$ from the $i$th part $\U_i$. If so, then we additionally send the index of $(b_i, r_{b_i})$ among all tuples $(b'_i, r'_{b'_i}) \in (A, \RA)_i^p$ for which $G$ answers $b'_i \in B$ when queried with $b'_i$.
        \item Finally, we send all remaining elements in $B$ explicitly. 
    \end{enumerate}

    Now, we sum up the size of the encoding: First, we get a contribution of $O(n)$ bits from the information required to specify which parts in $[m]$ $B$ samples from.
    Next, we spend other $O(n)$ bits to specify whether $b_i$ is in the plausible set, where $b_i$ represents the element in $B$ from the $i$th part of the universe.
    Then, when $b_i$ is in the plausible set, we must send the index of $b_i$ among all tuples in $(A, \RA)_i^p$ for which $G$ returns $B$.
    Recall that the probability that $b_i$ is in the plausible set, averaged over all $i$, is at least
    \begin{align*}
    \frac{(u/n) 2^v}{2^{S/m + O(1)}} \cdot \frac{1}{(u/n) - 1} = \frac{2^v}{2^{S/m + O(1)}}
    \end{align*}
    by \cref{clm:plausibleSize} and by the fact that $\PB$ is sampled uniformly and independently of $A$ and $\RA$. According to the extension of \cref{clm:returnB}, conditioned on $b_i$ is in the plausible set, the expected number of elements in the plausible set that lets $G$ return $B$ is bounded by $m + 1$.
    When a $b_i$ resides in the plausible set, we need $\log(m+1)$ bits to encode the index; otherwise, we need $\log(u/m)$ bits to send it explicitly in the third step.
    Thus, the expected number of bits we send for all $b_i$'s is
    \begin{align*}
    &\phantom{{}={}}\frac{n}{v} \bk*{\frac{2^v}{2^{S/m + O(1)}} \cdot \log(m+1) + \bk*{1 - \frac{2^v}{2^{S/m + O(1)}}} \cdot \log\bk*{\frac{u}{m}}} \\
    &= \frac{n}{v} \log \bk*{\frac{u}{m}} - \frac{n}{v} \cdot \frac{2^v}{2^{S/m + O(1)}} \cdot \log\bk*{\frac{u}{m^2}} + O(n).
    \end{align*}
    Combined with the $O(n)$ bits in earlier steps, we get the desired encoding length in the claim. Finally, once $B$ is recovered from the message, $\RB$ can also be recovered by querying the given data structure $G$.
\end{proof}

\begin{theorem}\label{thm:lowerBoundIncremental}
    Let $S \ge \max\BK{H(F), H(G)}$ be the expected size of the incremental retrieval data structure, and let $u \geq n^3$. Then,
    \[
    S = nv + \Omega \left (n \log \left (\frac{\log(n)}{v}\right ) \right).
    \]
\end{theorem}

\begin{proof}
From \cref{clm:protocolBound}, we know $H(G) \ge H(B) - H(B, \RB \mid F, G) + nv - O(n)$; we also showed that
\[
H(B, \RB \mid F, G) \le O(n) + \frac{n}{v} \log \bk*{\frac{u}{m}} - \frac{n}{v} \cdot \frac{2^v}{2^{S/m + O(1)}} \cdot \log\bk*{\frac{u}{m^2}}
\]
via an efficient encoding in \cref{clm:encodingSize}. Combined, it must be the case that
\begin{align*}
    S &\ge H(G) \ge H(B) + nv - \frac{n}{v} \log \bk*{\frac{u}{m}} + \frac{n}{v} \cdot \frac{2^{v}}{2^{S/m + O(1)}} \cdot \log \bk*{\frac{u}{m^2}} - O(n) \\
    &= \bk*{\frac{n}{v} \log \bk*{\frac{u}{m}} - O(n)} + nv - \frac{n}{v} \log \bk*{\frac{u}{m}} + \frac{n}{v} \cdot \frac{2^{v}}{2^{S/m + O(1)}} \cdot \log \bk*{\frac{u}{m^2}} - O(n) \\
    &\ge nv + \frac{n}{v} \cdot \frac{2^v}{2^{S/m + O(1)}} \cdot \log(n) - O(n), \numberthis\label{eq:incremental-LB-on-S}
\end{align*}
where the equality holds as $H(B) = (n/v) \log (u/m) - O(n)$; the last inequality holds as $\log(n) \le \log(u/m^2)$ due to our assumption of $u \ge n^3$. Next, we let $T = S - nv + O(n)$, so \eqref{eq:incremental-LB-on-S} becomes
\begin{align*}
    T &\ge \frac{n}{v} \cdot \frac{2^v}{2^{(T + nv + O(1))/m}} \cdot \log(n) \\
    &= \frac{n}{v} \cdot \frac{1}{2^{(T + nv - mv + O(1)) / m}} \cdot \log(n) \\
    &= \frac{n}{v} \cdot \frac{1}{2^{(T + n + O(1)) / m}} \cdot \log(n).
\end{align*}
Taking the logarithm, it yields that
\begin{align*}
    \log(T) + \frac{T + n + O(1)}{m} &\ge \log(n) - \log(v) + \log \log(n) \\
    \log(T/n) + \frac{T + n + O(1)}{m} &\ge \log\log(n) - \log(v) = \log\bk*{\frac{\log(n)}{v}} \\
    \Theta(T/n) &= \log\bk*{\frac{\log(n)}{v}}.
\end{align*}
The theorem follows as
\begin{align*}
    S = nv - O(n) + \Omega\bk*{n \log\bk*{\frac{\log(n)}{v}}}
    &= nv + \Omega\bk*{n \log\bk*{\frac{\log(n)}{v}}}.
    \qedhere
\end{align*}
\end{proof}

\begin{corollary}
    If $v = o(\log(n))$, then $S \geq nv + \omega(n)$.
\end{corollary}

\begin{proof}
    In \cref{thm:lowerBoundIncremental}, this yields $S = nv + \Omega\bk*{n \log\bk*{\frac{\log(n)}{o(\log(n))}}} = nv + \omega(n)$.
\end{proof}

\begin{remark}
    In particular, even if the values are $\log^{0.99}(n)$ bits, this bound implies that we will need $\Omega(n \log\log(n))$ extra wasted bits, meaning there is a sharp phase transition from when the values are $\log(n)$ bits, to slightly smaller than $\log(n)$ bits. 
\end{remark}

\begin{remark}
    Note that while the lower bound presented here does not provide anything meaningful when $v \geq \log(n)$, we still get a blanket $nv + \Omega(n)$ lower bound. This is because the work of \cite{KW24} showed a space lower bound of $nv + \Omega(n)$ bits in the value-dynamic setting, which supports a subset of the operations of the incremental setting, so their lower bound can be directly applied here. Thus, the overall bound reads
    \[
    nv + \Omega \left (n + n \log \left ( \frac{\log(n)}{v}\right ) \right ).
    \]
\end{remark}

\section{Lower Bound for Dynamic Retrieval Data Structures}
\label{sec:lb_dynamic}
In the previous sections, we have seen that in incremental retrieval data structures, there is a phase transition in the number of bits required for the data structure as the value sizes approach $\log(n)$ bits: When values are larger, we know that we can construct retrieval data structures that waste only $O(n)$ bits, but when values are smaller, we must necessarily waste $\omega(n)$ bits (tending towards $\Omega(n \log\log(n))$ bits). In this section, we show that for \emph{dynamic} retrieval data structures (where we allow deletions), there is no such phase transition, and for any choice of $v$, the data structure requires $nv + \Omega(n \log\log(n))$ bits. 

\newcommand{\RBo}{R_{B_0}}
\newcommand{\RBop}{R'_{B_0}}
\renewcommand{\Plausible}{(A, \RA, \RBo)_i^p}

\subsection{Lower Bound Set-up}

Now, we explain the set-up of our proof. 

\begin{enumerate}
    \item We let $\mathcal{U} = [u]$ denote the universe. We break $[u] - [n/2]$ into $n/2$ equal-sized parts denoted $\mathcal{U}_1, \dots, \mathcal{U}_{n/2}$, and let $A$ be a set of size $n/2$ by sampling one element from each part.
    \item Let $\RA$ be a set of $v$-bit values for $A$, such that each value has first bit $1$.
    \item Let $B_0 = \{1, 2, \ldots, n/2\}$, and let $\RBo$ be $v$-bit values for $B_0$.
    \item Construct a data structure $F$ by inserting $A$ with values $\RA$, and $B_0$ with values $\RBo$. 
    \item Now construct data structure $G$ by deleting $B_0$, and inserting $B$ with random $v$-bit values $\RB$, where $B$ is formed by sampling one random element from each of $\mathcal{U}_1, \dots, \mathcal{U}_{n/2}$, such that there is no intersection with $A$. Further, each value in $\RB$ must have first bit $0$. 
\end{enumerate}

As in the previous section, we view $\RA$ as a set of $n/2$ $v$-bit strings. We associate the $i$th $v$-bit string in $\RA$ with the $i$th smallest element in $A$. Because the values for elements in $A$ have first bit $1$, and $B$ have first bit $0$, we can use the data structure $G$ to distinguish between elements being in $A$ vs.~$B$.

By Yao's minimax principle, we also assume that the algorithms constructing $F$ and $G$ are \emph{deterministic}, denoted by functions $F = F(A, \RA, \RBo)$ and $G = G(F, B, \RB)$. Similar to the previous section, we still use $H(X)$ to denote the entropy of an object $X$, and thus $H(F)$ and $H(G)$ form lower bounds of the size of the retrieval data structure. To start the proof, we define a one-way communication protocol to lower bound the entropy of $G$.

\begin{claim}\label{clm:dynamicEncodingSetup}
    We must have $H(B \mid F, G) + H(G) \geq nv + H(B) - O(n)$.
\end{claim}

\begin{proof}
Consider the following one-way communication game. Alice wishes to send a single message to Bob, such that Bob can recover all of $A, \RA, B, \RB, \RBo$. The protocol is as follows.

\begin{enumerate}
\item Alice sends a single message consisting of $M = \BK{A, \RBo, G, (B \mid F, G)}$.
\item Given $G$ and $A$, Bob can recover $\RA$, as $A$ is still in the key set of the retrieval data structure $G$.
\item Based on the knowledge of $B_0$, $\RBo$, $A$, and $\RA$, Bob reconstructs the data structure $F$, which depends only on these parameters.
\item Now that Bob has $F$ and $G$, he can reconstruct $B$ with the help of $(B \mid F, G)$.
\item Finally, because Bob has $G$, and the set $B$ has been inserted into $G$ with associated values $\RB$, Bob can recover the values $\RB$.
\end{enumerate}

The entropy of the information that Bob recovers is at least 
\[
H(A, \RA, B, \RB, \RBo) = H(\RA) + H(\RB) + H(\RBo) + H(A) + H(B \mid A).
\]
For the terms on the right-hand side, we have $H(\RA) \geq (n/2)(v-1)$, $H(\RB) \geq (n/2)(v-1)$, and $H(B \mid A) \geq H(B) - O(n)$,
where the final equality holds because every element in $(B \mid A)$ is uniformly random from $(\frac{u}{n/2} - 2)$ possibilities, lowering the entropy by at most 1 bit compared to $B$ without conditioning on $A$. Thus, we see that the information Bob recovers has at least $nv + H(\RBo) + H(A) + H(B) - O(n)$ bits of entropy.

Finally, observe that the message Alice sends contains only $ A, \RBo, G, (B \mid F, G)$. Thus, it must be the case that 
\[
H(A) + H(\RBo) + H(G) + H(B \mid F, G) \geq nv + H(\RBo) + H(B) + H(A) - O(n),
\]
which means that 
\[
H(G) + H(B \mid F, G) \geq nv + H(B) - O(n).
\qedhere
\]
\end{proof}

\begin{remark}
    It suffices to prove that $H(B \mid F, G) \leq H(B) - \Omega(n \log\log(n))$, as this then immediately implies that $H(G) \geq nv + \Omega(n \log\log(n))$.
\end{remark}

To show that $H(B \mid F, G)$ is small, we will create an explicit encoding of $(B \mid F, G)$ which has a small expected size. This will immediately yield a bound on the entropy.

\subsection{Encoding Argument for \texorpdfstring{$(B \mid F, G)$}{(B | F, G)}}

In this subsection, we present a sequence of claims and definitions that we will use to conclude that $H(B \mid F, G) \leq H(B) - \Omega(n \log\log(n))$.
Similar to the incremental lower bound in \cref{sec:lb_incremental}, we start by defining the \emph{plausible sets} of a data structure $F$.

\begin{definition}
    We say that a tuple $(x, r_x, r_i)$ is in $(A, \RA, \RBo)_i^p$ if there exists $A', R'_{A'}, R'_{B_0}$ such that
    \begin{enumerate}
        \item $F(A, \RA, \RBo) = F(A', R'_{A'}, R'_{B_0})$,
        \item $x$ is the $i$th element of $A'$ with associated value $r_x \in R'_{A'}$, and
        \item $r_i \in R'_{B_0}$ is the $i$th value in $R'_{B_0}$.
    \end{enumerate}
\end{definition}

To clarify, when we say that $r_x \in R'_{A'}$, this means that the value associated to $x$ in $R'_{A'}$ is $r_x$.

\begin{claim}\label{clm:plausible_size_dynamic}
    Let $S \geq H(F(A, \RA, \RBo))$. Then,
    \[
    \E_{i, A, \RA, \RBo}\Bk[\big]{|(A, \RA, \RBo)_i^p|} \geq \frac{u - n/2}{n/2} \cdot \frac{1}{2^{(2S - 2nv)/n + O(1)}}.
    \]
\end{claim}

\begin{proof}
We prove this via an encoding argument. Suppose Alice selects a random set $A$, values $\RA$ and $\RBo$, creates the data structure $F(A, \RA, \RBo)$, and then sends this data structure $F$ to Bob. Once Bob receives $F$, he can compute the plausible set $(A, \RA, \RBo)_i^p$ (since this depends only on the data structure $F$). So, in order for Bob to recover $A, \RA, \RBo$, it suffices for Alice to additionally send the index of $(a_i, r_{a_i}, r_i)$ among all tuples in $(A, \RA, \RBo)_i^p$ for each $i \in [n/2]$, where $a_i$ denotes the $i$th element in $A$, $r_{a_i}$ denotes $a_i$'s associated value in $\RA$, and $r_i$ denotes the $i$th value in $\RBo$.

In this process, the total expected size of the message Alice sends is
\[
S + \sum_{i = 1}^{n/2} \E_{A, \RA, \RBo}\Bk[\big]{\log|(A, \RA, \RBo)_i^p|}. \numberthis \label{eq:alice_message}
\]
Bob recovers $A, \RA, \RBo$, which has entropy
\[
H(A, \RA, \RBo) = H(A) + H(\RA) + H(\RBo) \geq \frac{n}{2} \cdot \log \left ( \frac{u - n/2}{n/2} \right ) + nv - O(n). \numberthis\label{eq:bob_info}
\]
The expected length of Alice's message should be larger than the entropy, so \eqref{eq:alice_message} $\ge$ \eqref{eq:bob_info}. Dividing by $n/2$, we get
    \[
    \E_{i, A, \RA, \RBo}\Bk[\big]{\log |(A, \RA, \RBo)_i^p|} \geq \log \left ( \frac{u - n/2}{n/2} \right ) + 2v - 2S/n - O(1),
    \]
    and by Jensen's inequality,
    \[
     \E_{i, A, \RA, \RBo}\Bk[\big]{|(A, \RA, \RBo)_i^p|} \geq \frac{u - n/2}{n/2} \cdot \frac{1}{2^{(2S - 2nv)/n + O(1)}}.
     \qedhere
    \]
\end{proof}

Similar to before, the next step is to introduce a notion of feasibility which captures tuples that remain plausible for $A$ after $B$ has been revealed. 

\begin{definition}\label{def:feasible_dynamic}
    We say that a tuple $(x, r_x, r_i)$ \emph{remains feasible} for $(A, \RA, \RBo)_i^p$ after $B$, if there exists $A', R'_{A'}, R'_{B_0}$ such that
    \begin{enumerate}
        \item $F(A, \RA, \RBo) = F(A', R'_{A'}, R'_{B_0})$,
        \item $x$ is the $i$th element of $A'$ with associated value $r_x \in R'_{A'}$,
        \item $r_i \in R'_{B_0}$ is the $i$th value in $R'_{B_0}$, and
        \item $A' \cap B = \emptyset$.
    \end{enumerate}
\end{definition}

For a random choice of $B$, we show that in expectation, most elements in $(A, \RA, R_{B_0})_i^p$ remain feasible after $B$.

\begin{claim}\label{clm:dynamicProbabilityFeasible}
    Let $A, \RA, \RBo$ be fixed, and consider an element $(x, r_x, r_i) \in (A, \RA, \RBo)_i^p$. Over a random choice of the set $B$, $(x, r_x, r_i)$ remains feasible for $(A, \RA, \RBo)_i^p$ with probability
    $1 - \frac{n^2}{2u}$.
\end{claim}

\begin{proof}
    $(x, r_x, r_i) \in (A, \RA, \RBo)_i^p$ indicates that there exists $A', R'_{A'}, R'_{B_0}$ satisfying the first three conditions in \cref{def:feasible_dynamic}; if $A' \cap B = \emptyset$, this suffices to conclude that $(x, r_x, r_i)$ remains feasible. To bound the probability that $A' \cap B \neq \emptyset$, we see that $B$ randomly samples elements from the $n/2$ parts of the domain. Thus, the probability that $B$ intersects with $A'$ is at most
    \[
    \frac{n}{2} \cdot \frac{1}{\frac{u}{n / 2} - 2} 
    \le \frac{n}{2} \cdot \frac{n}{u}
    = \frac{n^2}{2u}.
    \qedhere
    \]
\end{proof}

Similar to \cref{sec:lb_incremental}, \cref{clm:dynamicProbabilityFeasible} also holds conditioned on that the $i$th element of $B$ equals a given element $b_i \ne x$; when $b_i = x$, it is impossible to let the tuple $(x, r_x, r_i)$ remain feasible after $B$.

\begin{claim}\label{clm:dynamicExpectedFeasible}
    Let $A, \RA, \RBo$ be fixed. The expected number of tuples $(x, r_x, r_i) \in (A,\RA, \RBo)_i^p$ which do not remain feasible for $A$ after a random choice of $B$ is at most $n + 1$.
    Furthermore, this holds even if the $i$th element in $B$ is arbitrarily given while other elements in $B$ are sampled randomly.
\end{claim}

\begin{proof}
First, note that there cannot be two tuples sharing the same $x$ in $(A,\RA, \RBo)_i^p$, because querying the data structure $F$ with $x$ and $i$ will determine the associated values $r_x \in \RA$ and $r_i \in \RBo$ uniquely.
This implies that the number of possible tuples in $(A, \RA, \RBo)_i^p$ is at most $|\mathcal{U}_i| = \frac{u - n/2}{n/2} = \frac{2u}{n} - 1$.

Then, over a random choice of $B$, we know that each such tuple remains feasible for $(A, \RA, \RBo)_i^p$ after $B$ with probability $1 - \frac{n^2}{2u}$ due to \cref{clm:dynamicProbabilityFeasible}. Thus, the number of tuples which \emph{do not} remain feasible for $(A, \RA, \RBo)_i^p$ is bounded in expectation by $\bk[\big]{\frac{2u}{n} - 1} \cdot \frac{n^2}{2u} \le n$.

When the $i$th element $b_i \in B$ is arbitrarily given, there are two cases: For tuples $(x, r_x, r_i)$ with $x \ne b_i$, the probability of not remaining feasible is still at most $\frac{n^2}{2u}$; the tuple with $x = b_i$ cannot remain feasible, but there is only one such tuple. Combined, the expected number of tuples not remaining feasible is at most $\bk[\big]{\frac{2u}{n} - 1} \cdot \frac{n^2}{2u} + 1 \le n + 1$.
\end{proof}

\begin{claim}\label{clm:dynamicMembershipFeasibility}
    If a tuple $(x, r_x, r_i)$ is in $(A, \RA, \RBo)_i^p$ and remains feasible after $B$,
    then $G$ \emph{must} report that $x \in A$.
\end{claim}

\begin{proof}
    Observe that we could have initialized the data structure with the $A', R'_{A'}, R'_{B_0}$ that witness $(x, r_x, r_i) \in (A, \RA, \RBo)_i^p$.
    By definition, the result data structure has the same encoding as $F$, and still maintains the property of disjointness between $A$ and $B$, as required by our set-up. Because this case is indistinguishable to $G$ from the case when we initialized $F$ with $(A, \RA, \RBo)$, $G$ must treat both cases as possible, and therefore respond $A$ to any query from $x$ where $(x, r_x, r_i)$ is a tuple which remains feasible.
\end{proof}

Combining \cref{clm:dynamicExpectedFeasible,clm:dynamicMembershipFeasibility}, we know that the expected number of elements $(x, r_x, r_i) \in \Plausible$ for which $G$ answers $x \in B$ is at most $n + 1$.
This also holds conditioned on that a given element $b_i$ is the $i$th element of $B$; further, by the law of total probability, it holds conditioned on the $i$th element $b_i$ of $B$ being in the plausible set.
Based on this fact, we create an encoding of the set $B$ conditioning on $F$ and $G$ as shown below.

\begin{claim}\label{clm:dynamicEncoding}
    Let $S \geq \max\BK{H(F), H(G)}$. Then, there exists an encoding of $B$ given $F, G$ of expected size
    \[
    H(B) - \frac{n}{2} \cdot \frac{\log(n)}{2^{(2S - 2nv)/n +  O(1)}} + O(n).
    \]
\end{claim}

\begin{proof}
Let the set $B$ be sampled at random under the established rules. Then our encoding is as follows:
First, we send $O(n)$ bits indicating for each $b_i \in B$ if there exists $r_{b_i}, r_i$ such that $(b_i, r_{b_i}, r_i) \in (A, \RA, \RBo)_i^p$ (i.e., has the same encoding as $F$).
When this occurs, we also send the index of $b_i$ among all such elements in $(A, \RA, \RBo)_i^p$ for which $G$ returns $B$.
For the other remaining $b_i$, we simply send the description of $b_i$ exactly.

Now, let us calculate the expected size of this encoding. The first part of the message indicating whether each $b_i$ is in the plausible set takes $O(n)$ bits to send. Next, recall \cref{clm:plausible_size_dynamic} that the size of the plausible set satisfies
\[\E_{i, A,\RA, \RBo}\Bk[\big]{|(A,\RA, \RBo)_i^p|} \geq \frac{u - n/2}{n/2} \cdot \frac{1}{2^{(2S - 2nv)/n + O(1)}}.\]
Thus, for a random choice of element $b_i$, the probability that it lies in the plausible set is at least
$1/2^{(2S - 2nv)/n + O(1)}$. When the element $b_i$ is in the plausible set, we send the index among all $x$ which lie in the plausible set yet $G$ responds with ``$B$'' when queried with $x$.
As pointed out above, the number of such elements is bounded in expectation by $n + 1$, conditioned on that $b_i$ is in the plausible set.
Thus, the index for each $b_i$ in the plausible set takes $\log(n + 1)$ bits to send. Finally, for each $b_i$ not in the plausible set, we send it explicitly, spending $\log\bk[\big]{\frac{u - n/2}{n/2}}$ bits each. The total encoding length is
\begin{align*}
    &\phantom{{}={}} O(n) + \frac{n}{2} \cdot \frac{1}{2^{(2S - 2nv) / n + O(1)}} \cdot \log(n + 1) + \frac{n}{2} \cdot \bk*{1 - \frac{1}{2^{(2S - 2nv) / n + O(1)}}} \cdot \log\bk*{\frac{u - n/2}{n/2}} \\
    &= O(n) + \frac{n}{2} \log\bk*{\frac{u - n/2}{n/2}} - \frac{n}{2} \cdot \frac{1}{2^{(2S - 2nv) / n + O(1)}} \cdot \log \bk*{\frac{u - n/2}{(n/2) \cdot (n + 1)}} \\
    &\le O(n) + H(B) - \frac{n}{2} \cdot \frac{1}{2^{(2S - 2nv) / n + O(1)}} \cdot \log(n),
\end{align*}
where the last step is because $\log\bk[\big]{\frac{u - n/2}{(n/2) \cdot (n + 1)}} \ge \log\bk[\big]{\frac{u}{n^2}} - 1 \ge \log(n) - 1$ due to our assumption of $u \ge n^3$.
\end{proof}

Rearranging the terms in \cref{clm:dynamicEncoding}, we know
\[
H(B) - H(B \mid F, G) \geq \frac{n}{2} \cdot \frac{\log(n)}{2^{(2S - 2nv)/n + O(1)}} - O(n).
\numberthis \label{eq:dynamicEncodingCor}
\]
This combined with \cref{clm:dynamicEncodingSetup} leads to our final theorem.

\subsection{Final Statement}

\begin{theorem}
    Let $S = \max\BK{H(F), H(G)}$, and let $v$ be arbitrary. Then, $S = nv + \Omega(n\log\log(n))$.
\end{theorem}

\begin{proof}
    We know from \cref{clm:dynamicEncodingSetup} that $S \geq nv + H(B) - H(B\mid F, G)$. From \cref{clm:dynamicEncoding}, we know \eqref{eq:dynamicEncodingCor}.
    Combined,
    \[
    S \geq nv + \frac{n}{2} \cdot \frac{1}{2^{(2S - 2nv)/n + O(1)}} \log(n) - O(n).
    \]
    Setting $T = S - nv + O(n)$, this means that
    \[
    T \geq \frac{n}{2} \cdot \frac{1}{2^{2T/n + O(1)}} \log(n).
    \]
    Taking the logarithm on both sides, we get
    \[
    \frac{2T}{n} + \log(2T/n) \geq \log\log(n) - O(1),
    \]
    which means that
    \[
    O(T/n) \geq \log\log(n),
    \]
    therefore $T = \Omega(n\log\log(n))$. Equivalently, $S = nv + \Omega(n\log\log(n))$.
\end{proof}

Because we have a lower bound on the maximum entropy of $F, G$ of $nv + \Omega(n\log\log(n))$ bits, this means that the retrieval data structure requires this many bits of space in expectation as well, giving us our desired space lower bound.

\section{Open Questions}
\label{sec:open}
We conclude the paper with several related open questions.

\paragraph{Retrieval with dynamic resizing.} The main open question left by \cite{DHPP06} was whether one can support dynamic resizing, meaning that the space bound on the data structure depends on the current size of the key set, instead of a fixed upper bound of the size. To the best of our knowledge, this resizing question remains completely open in both the dynamic and incremental cases.

For the incremental setting, one further question is whether resizable retrieval exhibits the same type of phase transition our in the current paper. Note that the incremental algorithm in this paper heavily relies on the fixed upper bound $n$ on the number of keys, and thus cannot be directly extended to support dynamic resizing. 

\paragraph{Practical algorithms for retrieval.} The second question is to design dynamic retrieval data structures that can be efficiently implemented for practical applications.
Despite the theoretical progress on retrieval data structures, there have not yet (as far as we know) been any practical implementations of dynamic retrieval. It seems likely that the known algorithms are not simple enough to be efficient in practice. If one could make dynamic retrieval practical, it is likely that it would find many applications to real-world space-efficient data-structure design.

\bibliographystyle{alpha}
\bibliography{ref}

\appendix
\section{Insertion-Only Linear Probing}
\label{sec:insertion_only_linear_probing}
In this section, we construct an efficient insertion-only linear-probing hash table that was used in our upper bound result in \cref{sec:hashtable}. We first recall \cref{thm:linearProbing}.

\LinearProbing*

\begin{proof}[Proof of \cref{thm:linearProbing}]

    We allocate an array of $M = m \bk[\big]{1 + \frac{1}{\log^2(n)}}$ slots, each capable of storing a $v$-bit value.
    When a key-value pair $(k, y)$ is inserted, we use a \emph{5-wise independent} hash function $h$ to map $k$ to a slot $h(k) \in [M]$, and follow the linear-probing policy to find the first empty slot after $h(k)$ to store the value $y$, then return the offset $p_k$.

    Assume the $i$th inserted key-value pair is $(k_i, y_i)$ for $i \in [m]$.
    For the $i$th inserted key $k_i$, according to \cite{pagh2011linear}, the expected offset $p_{k_i}$ is at most
    \[
    \E[p_{k_i}] = O\bk*{\bk*{\frac{M}{M - i}}^{2.5}}.
    \numberthis \label{eq:offset_expectation}
    \]
    As a loose bound, for all $i \le m$, there is $\E[p_{k_i}] \le O(\log^{5}(n))$, as the maximum load factor is set to $1 - \Theta(1 / \log^2(n))$.
    Taking a summation over $i \le m$, we obtain
    \[
    \E\Bk*{\sum_{i=1}^m p_{k_i}} \le O(m \log^5(n)).
    \]
    By Markov's inequality, we know that the expected number of offsets exceeding $\log^{100}(n)$ is bounded by $O(m / \log^{95}(n))$, satisfying one requirement of \cref{thm:linearProbing}.

    On the other hand, taking the logarithm of \eqref{eq:offset_expectation} and taking a summation over $i \le m$ will give
    \begin{align*}
        \E\Bk*{\sum_{i=1}^{m} \log p_{k_i}} &= O\bk*{\sum_{i=1}^{m} (-2.5) \log\bk*{1 - \frac{i}{M}}}
        = O\bk*{m \int_{\log^{-2}(m)}^{1} (- \ln t) \d t}
        = O(m),
    \end{align*}
    as desired.

    Next, we show how to achieve constant-expected-amortized-time operations. We maintain a bit vector of length $M$ to represent if each slot is empty. In addition, we divide the bit vector into blocks of size $C \defeq \log^{10}(n)$; for each block, we build a B-tree with branching factor $\sqrt{\log(n)}$, supporting the following two operations:
    \begin{enumerate}
    \item Given a position $s$ within the block, find the first 1 at or after $s$ (or report that there is no more 1 within the block).
    \item Update a bit from 1 to 0.
    \end{enumerate}
    The B-tree only has $O(1)$ levels, and the information on each node is just a bit vector of length $O(\sqrt{\log(n)})$ indicating if there are 1's in each child. By standard bit operations, one can see that both operations can be performed in worst-case $O(1)$ time. The B-tree only consumes constant bits for each bit it maintains, i.e., $O(1)$ bits per slot in the linear-probing hash table, which sums up to $O(M)$ and is acceptable.

    With the help of the B-trees, one can detect in constant time whether there are empty slots within an interval of $O(\log^{10}(n))$ slots; if there are empty slots, one can also find the first empty slot after $h(k)$ -- the slot that key $k$ hashes to -- within $O(1)$ time, thus completing the insertion to the hash table.
    When an insertion produces an offset $p_{k_i}$, the above insertion algorithm takes $O\bk[\big]{1 + \frac{p_{k_i}}{\log^{10}(n)}}$ time. Combined with the fact that $\E[p_{k_i}] \le O(\log^5(n))$, each insertion takes constant time in expectation.

    Finally, this linear-probing hash table is space-efficient: The $M$ slots consume $Mv = mv + o(m)$ bits of space, while the B-trees maintaining bit vectors cost $O(M) = O(m)$ bits of space. The hash function, which is only 5-wise independent, consumes $O(\log \U)$ bits which is negligible. They fit in the space requirement of the theorem.
\end{proof}

\end{document}